\documentclass[a4paper]{article}

\usepackage[english]{babel}
\usepackage{german}
\usepackage{amssymb,amsfonts,amsmath,amsthm,cite}
\usepackage[latin1]{inputenc}
\usepackage{fullpage}
\usepackage{graphicx}
\usepackage{subfigure}

\usepackage{float}
\usepackage{algorithm}
\usepackage[noend]{algorithmic}
\usepackage{color,hyperref}

\newcommand{\SWTVC}{sliding $\Delta$-window temporal vertex cover}
\newcommand{\SWTVCD}[1]{sliding #1-window temporal vertex cover}

\newcommand{\SWTVCProblem}{\textsc{Sliding Window Temporal Vertex Cover}}
\newcommand{\SWTVCProblemShort}{\textsc{SW-TVC}}

\newcommand{\DTVCproblem}{\textsc{$\Delta$-TVC}}
\newcommand{\DTVCproblemD}[1]{\textsc{#1-TVC}}

\newtheorem{theorem}{Theorem}

\newtheorem{corollary}{Corollary}

\newtheorem{definition}{Definition}
\newtheorem{observation}{Observation}

\newtheorem{lemma}{Lemma}

\newtheorem{problem}{Problem}

\newtheorem*{ethypot}{Exponential Time Hypothesis}

\begin{document}
\title{\vspace{-0.5cm}Temporal Vertex Cover with a Sliding Time Window\thanks{This work was partially supported 
by the NeST initiative of the School of EEE and CS at the University of Liverpool and 
by the EPSRC Grants EP/P020372/1 and EP/P02002X/1. 
A preliminary conference version of this work appeared in the Proceedings of ICALP 2018~\cite{AMSZ18}.}}
\author{Eleni C.~Akrida\thanks{Department of Computer Science, University of Liverpool, Liverpool, UK. 
Email: \texttt{e.akrida@liverpool.ac.uk}} 
\and George B.~Mertzios\thanks{Department of Computer Science, Durham University, Durham, UK. 
Email: \texttt{george.mertzios@durham.ac.uk}} 
\and Paul G.~Spirakis\thanks{Department of Computer Science, University of Liverpool, Liverpool, UK. 
Email: \texttt{p.spirakis@liverpool.ac.uk}}
\and Viktor Zamaraev\thanks{Department of Computer Science, Durham University, Durham, UK. 
Email: \texttt{viktor.zamaraev@durham.ac.uk}}}
\date{\vspace{-1.0cm}}
\maketitle

\begin{abstract}		
Modern, inherently dynamic systems are usually characterized by a network structure, i.e.~an underlying graph topology, 
which is subject to discrete changes over time. 
Given a static underlying graph $G$, a temporal graph can be represented via an assignment of a set 
of integer time-labels to every edge of $G$, indicating the discrete time steps when this edge is active.
While most of the recent theoretical research on temporal graphs has focused on the notion of a temporal path 
and other ``path-related'' temporal notions, only few attempts have been made to investigate ``non-path'' temporal graph problems. 
In this paper, motivated by applications in sensor and in transportation networks, we introduce and study 
two natural temporal extensions of the classical problem \textsc{Vertex Cover}. 
In both cases we wish to minimize the total number of ``vertex appearances'' that are needed 
to ``cover'' the whole temporal graph. 
In our first problem, \textsc{Temporal Vertex Cover}, the aim is to cover every edge at least once 
during the lifetime of the temporal graph, where an edge can be covered by one of its endpoints, only at 
a time step when it is active. 
In our second, more pragmatic variation \textsc{Sliding Window Temporal Vertex Cover}, we are also given 
a natural number $\Delta$, and our aim is to cover every edge at \emph{least once} 
at \emph{every $\Delta$ consecutive} time steps. 
We present a thorough investigation of the computational complexity and approximability of these two 
temporal covering problems. 
In particular, we provide strong hardness results, complemented by various approximation and exact 
algorithms. Some of our algorithms are polynomial-time, while others are asymptotically almost optimal
under the Exponential Time Hypothesis (ETH) and other plausible complexity assumptions.\newline
 
\end{abstract}

\section{Introduction and Motivation} \label{intro}

A great variety of both modern and traditional networks are inherently dynamic,
in the sense that their link availability varies over time. 
Information and communication networks, social networks, transportation networks, and several
physical systems are only a few examples of networks that change over time~\cite{Holme-Saramaki-book-13,michailCACM}. 
The common characteristic in all these application areas is that the network structure, 
i.e.~the underlying graph topology, is subject to \emph{discrete changes over time}.
In this paper we adopt a simple and natural model for time-varying networks which is given with 
time-labels on the edges of a graph, while the vertex set remains unchanged. 
This formalism originates in the foundational work of Kempe et al.~\cite{kempe}.

\begin{definition}[temporal graph]
\label{temp-graph-def} A \emph{temporal graph} is a pair $(G,\lambda)$,
where $G=(V,E)$ is an underlying (static) graph and $\lambda :E\rightarrow
2^{\mathbb{N}}$ is a \emph{time-labeling} function which assigns to every
edge of $G$ a set of discrete-time labels.
\end{definition}

For every edge $e\in E$ in the underlying graph $G$ of a temporal graph 
$(G,\lambda)$, $\lambda (e)$ denotes the set of time slots at which $e$ is \emph{active} in $(G,\lambda)$. 
Due to its vast applicability in many areas, this notion of temporal graphs 
has been studied from different perspectives under various names 
such as \emph{time-varying}~\cite{FlocchiniMS09,TangMML10-ACM,krizanc1}, 
\emph{evolving}~\cite{xuan,Ferreira-MANETS-04,clementi}, 
\emph{dynamic}~\cite{GiakkoupisSS14,CasteigtsFloccini12}, 
and \emph{graphs over time} \cite{Leskovec-Kleinberg-Faloutsos07}; 
for a recent attempt to integrate  existing models, concepts, and results 
from the distributed computing perspective see the 
survey papers~\cite{CasteigtsFloccini12,flocchini1,flocchini2} and the references therein. 
Data analytics on temporal networks have also been very recently studied
in the context of summarizing networks that represent sports teams' activity
data to discover recurring strategies and understand team
tactics~\cite{gionis1}, as well as extracting patterns from interactions between
groups of entities in a social network~\cite{gionis2}.

Motivated by the fact that information in temporal graphs can
``flow'' only along sequences of edges whose time-labels are increasing, 
most temporal graph parameters and optimization problems that have
been studied so far are based on the notion of temporal paths and other
``path-related'' notions, such as temporal analogues of distance, diameter,
reachability, exploration, and centrality~\cite{AkridaGMS16,erlebach,mertziosMCS19,michailTSP,AkridaGMS-TOCS17,enright2018deleting,AkridaMNRPZ_Stochastic-19,AkridaMS19}. 
In contrast, only few attempts have been made to define ``non-path'' temporal graph problems. 
Motivated by the contact patterns among high-school students, 
Viard et al.~\cite{viardClique, viardCliqueTCS}, and later Himmel et al.~\cite{neidermeier}, 
introduced and studied $\Delta$-cliques, an extension of the concept of cliques to temporal graphs, 
in which all vertices interact with each other at least once every $\Delta$ consecutive time steps 
within a given time interval. 
Furthermore, 
natural temporal extensions have been recently introduced and studied for the classical problems graph coloring~\cite{MMZ19} and maximum matching~\cite{baste2018temporal,MertziosMNZZ-Arxiv19}.

In this paper we introduce and study two natural temporal extensions of the problem \textsc{Vertex Cover} in static graphs, 
which take into account the dynamic nature of the network. 
In the first and simpler of these extensions, namely \textsc{Temporal Vertex Cover} (for short, \textsc{TVC}), 
every edge $e$ has to be ``covered'' at least once during the lifetime $T$ of the network (by one of its endpoints), 
and this must happen at a time step $t$ when $e$ is active. 
The goal is then to cover all edges with the minimum total number of such ``vertex appearances''. 
On the other hand, in many real-world applications where 
scalability is important, the lifetime~$T$ can be arbitrarily large but the network still needs to remain sufficiently covered. 
In such cases, as well as in safety-critical systems (e.g.~in military applications), 
it may not be satisfactory enough that an edge is covered just \emph{once} during the \emph{whole lifetime} of the network. 
Instead, every edge must be covered at least once within \emph{every small $\Delta$-window} of time in which it is active (for an appropriate value of $\Delta$), regardless of how large the lifetime is; 
this gives rise to our second optimization problem, namely \textsc{Sliding Window Temporal Vertex Cover} 
(for short, \textsc{SW-TVC}). 
Formal definitions of our problems \textsc{TVC} and 
\textsc{SW-TVC} are given in Section~\ref{preliminaries}. 
Here it is worth mentioning that very recently another temporal version of \textsc{Vertex Cover}~--namely 
the \emph{network-untangling} problem--~has been introduced and studied, motivated by applications 
in discovering events' timelines from complex interactions among entities~\cite{Rozenshtein17}.


Our main motivation for introducing and studying \textsc{TVC} and \textsc{SW-TVC} is of theoretical nature, 
namely to lift one of the most classical optimization problems, such as \textsc{Vertex Cover}, 
to the temporal setting. However, these temporal extensions of \textsc{Vertex Cover} 
could potentially also prove useful in extending the known practical applications 
of the static \textsc{Vertex Cover} problem. 
One example of such a possible application comes from the field of sensor networks, 
where several works considered problems of placing sensors to cover a whole area 
or multiple critical locations, e.g.~for reasons of surveillance. 
Such studies usually wish to minimize the number of sensors used or the total energy
required~\cite{kranakis1,kranakis2,kranakis3,zhu,nikoletseas2}.

\subsection{Our contribution}
\label{contributions}

In this paper we present a thorough investigation of the complexity and approximability of 
the problems \textsc{Temporal Vertex Cover} (\textsc{TVC}) and \textsc{Sliding Window Temporal Vertex Cover}
(\textsc{SW-TVC}) on temporal graphs.
We first prove in Section~\ref{sec:hardness-approx-TVC} that \textsc{TVC} remains NP-complete even 
on the special case of star temporal graphs, i.e.~when the underlying graph $G$ is a star. 
This NP-completeness is proved via a reduction from \textsc{Set Cover}, which provides a natural 
one-to-one correspondence between the input \textsc{Set Cover} instance and the produced 
instance of \textsc{TVC} on star temporal graphs. 
Furthermore we prove that, for any $\varepsilon<1$, \textsc{TVC} on star temporal graphs cannot be optimally solved in $O(2^{\varepsilon T})$ time, assuming the Strong Exponential Time Hypothesis (SETH), 
as well as that it does not admit a polynomial-time $(1-\varepsilon)\ln{n}$-approximation algorithm, 
unless NP has $n^{O(\log \log n)}$-time deterministic algorithms. 
On the positive side we prove that, on general temporal graphs with $n$ vertices, 
\textsc{TVC} can be $(H_{n-1}-\frac{1}{2})$-approximated in polynomial time, 
where $H_{n}=\sum_{i=1}^{n}\frac{1}{i} \approx \ln n$ is the $n$th harmonic number.

In Section~\ref{sec:tight-algorithms-SW-TVC} and in the reminder of the paper we deal with \textsc{SW-TVC}. 
We prove in Section~\ref{subsec:lower-bound-ETH} a strong complexity lower bound on arbitrary temporal graphs. 
More specifically we prove that, for \emph{any} (arbitrarily growing) 
functions $f:\mathbb{N}\rightarrow \mathbb{N}$ and $g:\mathbb{N}\rightarrow \mathbb{N}$, there 
exists a constant $\varepsilon \in (0,1)$ such that \textsc{SW-TVC} cannot be solved 
in $f(T)\cdot 2^{\varepsilon n \cdot g(\Delta)}$ time, assuming the Exponential Time Hypothesis (ETH). 
This ETH-based lower bound turns out to be asymptotically almost tight, as 
we present an exact dynamic programming algorithm with running time $O(T\Delta (n+m)\cdot 2^{n(\Delta+1)})$. 
This worst-case running time can be significantly improved in certain special temporal graph classes. 
In particular, when the ``snapshot'' of $(G,\lambda)$ at every time step has vertex cover number bounded by~$k$, 
the running time becomes $O(T\Delta (n+m)\cdot n^{k(\Delta+1)})$. 
That is, whenever $\Delta$ is constant, 
this algorithm is polynomial in the input size on temporal graphs with bounded vertex cover number 
at every time step. 
Notably, when every snapshot is a star (i.e.~a superclass of the star temporal graphs 
studied in Section~\ref{sec:hardness-approx-TVC}) the running time of the algorithm is $O(T\Delta (n+m)\cdot 2^{\Delta})$.

In Section~\ref{sec:hardDTVC} we prove strong inapproximability results for \textsc{SW-TVC}, 
even in the special case where the length of the sliding window is $\Delta=2$. 
In particular, we prove that, unless P=NP, this problem does not admit a \emph{Polynomial Time Approximation Scheme (PTAS)} 
even when $\Delta=2$, the maximum degree in the underlying graph $G$ is at most $3$, and every connected component of every snapshot has at most $7$ vertices. 
Finally, in Section~\ref{sec:approx} we provide a series of approximation algorithms for the general 
\textsc{SW-TVC} problem, with respect to various incomparable temporal graph parameters. 
In particular, we provide polynomial-time approximation algorithms with approximation ratios 
(i)~$\ln{n} + \ln{\Delta} + \frac{1}{2}$, 
(ii)~$2k$, where $k$ is the maximum number of times that each edge can appear in a sliding $\Delta$ time window (thus implying a ratio of $2\Delta$ in the general case), 
(iii)~$d$, where $d$ is the maximum vertex degree at every snapshot of $(G,\lambda)$. 
Note that, for $d=1$, the latter result implies that \textsc{SW-TVC} can be optimally solved in polynomial time whenever every snapshot of $(G,\lambda)$ is a matching.

\section{Preliminaries and notation}
\label{preliminaries}

A theorem proving that a problem is NP-hard does not provide much information 
about how efficiently (although not polynomially, unless P=NP) this problem can be solved. 
In order to prove some useful complexity lower bounds, we mostly need to rely on some 
complexity hypothesis that is stronger than``P~$\neq$~NP''. 
Impagliazzo, Paturi, and Zane formulated the \emph{Exponential Time Hypothesis (ETH)}~\cite{ETH01}, 
which is one of the established and most well-known complexity hypotheses.

\begin{ethypot}[ETH~\cite{ETH01}]
\label{ETH-def}
There exists an $\varepsilon<1$ such that \textsc{3SAT} cannot be solved in $O(2^{\varepsilon n})$ time, where $n$ is the number of variables in the input 3-CNF formula.
\end{ethypot}

In addition to formulating ETH, Impagliazzo and Paturi proved the celebrated \emph{Sparsification Lemma}~\cite{Sparsification01}, which has the following theorem as a consequence. This result is quite useful 
for providing lower bounds assuming ETH, as it expresses the running time in terms of the size of the input 
3-CNF formula, rather than only the number of its variables.

\begin{theorem}[\hspace{-0,001cm}\cite{Sparsification01}]
\label{ETH-thm}
\textsc{3SAT} can be solved in time $2^{o(n)}$ if and only if it can be solved in time~$2^{o(m)}$
on 3-CNF formulas with $n$ variables and $m$ clauses.
\end{theorem}

Given a (static) graph $G$, we denote by $V(G)$ and $E(G)$ the sets of its
vertices and edges, respectively. An edge between two vertices $u$ and $v$
of $G$ is denoted by $uv$, and in this case $u$ and $v$ are said to be \emph{%
adjacent} in $G$. 
We consider here temporal graphs with a finite set of integer labels assigned to every edge;
the maximum label assigned by $\lambda $ to an edge of~$G$, called the \emph{lifetime} of $(G,\lambda )$, is denoted by $%
T(G,\lambda )$, or simply by $T$ when no confusion arises. 
That is, $T(G,\lambda )=\max \{t\in \lambda (e):e\in E\}$ and $T$ is \emph{finite}. For every $i,j\in \mathbb{N}
$, where $i\leq j$, we denote $[i,j]=\{i,i+1,\ldots ,j\}$. Throughout the paper  we refer to each
integer $t\in \lbrack 1,T]$ as a \emph{time point} (or \emph{time slot}) of $%
(G,\lambda )$. The \emph{instance} (or \emph{snapshot}) of $(G,\lambda )$ 
\emph{at time }$t$ is the static graph $G_{t}=(V,E_{t})$, where $%
E_{t}=\{e\in E:t\in \lambda (e)\}$. For every $i,j\in \lbrack 1,T]$, where $%
i\leq j$, we denote by $(G,\lambda )|_{[i,j]}$ the restriction of $%
(G,\lambda )$ to the time slots $i,i+1,\ldots ,j$, i.e.~$(G,\lambda
)|_{[i,j]}$ is the sequence of the instances $G_{i},G_{i+1},\ldots ,G_{j}$.
We assume in the remainder of the paper that every edge of $G$ appears in at
least one time slot in $\{1,...,T\}$, namely $\bigcup_{t=1}^{T}E_{t}=E$.

Although some optimization problems on temporal graphs may be hard to solve
in the worst case, an optimal solution may be efficiently computable when
the input temporal graph $(G,\lambda )$ has special properties, i.e.~if $%
(G,\lambda )$ belongs to a special \emph{temporal graph class} (or \emph{%
time-varying graph class}~\cite{flocchini1,CasteigtsFloccini12}). To specify
a temporal graph class we can restrict (a) the \emph{underlying topology} $G$, 
or (b) the \emph{time-labeling} $\lambda $, i.e.~the temporal pattern in
which the time-labels appear, or both. 

\begin{definition}
\label{temporal-classes-def}
Let $(G,\lambda )$ be a temporal graph and let $\mathcal{X}$ be a class of (static) graphs. 
If~$G\in \mathcal{X}$ then $(G,\lambda)$ is an \emph{$\mathcal{X}$ temporal graph}. 
On the other hand, if $G_{i}\in \mathcal{X}$ for every $i\in \lbrack 1,T]$,
then $(G,\lambda )$ is an \emph{always $\mathcal{X}$ temporal graph}.
\end{definition}

In the remainder of the paper we denote by $n=|V|$ and $m=|E|$ the number of
vertices and edges of the underlying graph $G$, respectively, unless
otherwise stated. Furthermore, unless otherwise stated, we assume that the
labeling $\lambda $ is arbitrary, i.e.~$(G,\lambda )$ is given with an
explicit list of labels for every edge. That is, the \emph{size} of the
input temporal graph $(G,\lambda )$ is $O\left(
|V|+\sum_{t=1}^{T}|E_{t}|\right) =O(n+mT)$. In other cases, where $\lambda $
is more restricted, e.g.~if $\lambda $ is periodic or follows another
specific temporal pattern, there may exist a more succinct representations
of the input temporal graph.

For every $u\in V$ and every time slot $t$, we denote the \emph{appearance
of vertex} $u$ \emph{at time} $t$ by the pair $(u,t)$. That is, every vertex 
$u$ has $T$ different appearances (one for each time slot) during the
lifetime of $(G,\lambda )$. Similarly, for every vertex subset $S\subseteq V$
and every time slot $t$ we denote the \emph{appearance of set} $S$ \emph{at
time} $t$ by $(S,t)$. With a slight abuse of notation, we write $(S,t)=\bigcup
_{v\in S}(v,t)$. A \emph{temporal vertex subset} of $(G,\lambda )$ is a set $%
\mathcal{S}\subseteq \{(v,t):v\in V,1\leq t\leq T\}$ of vertex appearances
in $(G,\lambda )$. Given a temporal vertex subset $\mathcal{S}$, for every
time slot $t\in \lbrack 1,T]$ we denote by $\mathcal{S}_{t}=\{(v,t) : (v,t)\in \mathcal{S}\}$ the set of
all vertex appearances in $\mathcal{S}$ at the time slot $t$. 
Similarly, for any pair of time slots $i,j\in [1,T]$, where $i\leq j$, $\mathcal{S}|_{[i,j]}$ is the 
restriction of the vertex appearances of $\mathcal{S}$ within the time slots $i,i+1, \ldots, j$. 
Note that the \emph{cardinality} of the temporal vertex subset $\mathcal{S}$ is $|\mathcal{S}%
|=\sum_{1\leq t\leq T}|\mathcal{S}_{t}|$.

\subsection{\label{TVC-subsec}Temporal Vertex Cover}

Let $\mathcal{S}$ be a temporal vertex subset of $(G,\lambda )$. Let $%
e=uv\in E$ be an edge of the underlying graph $G$ and let $(w,t)$ be a
vertex appearance in $\mathcal{S}$. We say that vertex $w$ \emph{covers} the edge $e$ if $w\in
\{u,v\}$, i.e.~$w$ is an endpoint of $e$; in that case, edge $e$ is \emph{%
covered} by vertex $w$. Furthermore we say that the vertex appearance $(w,t)$
\emph{temporally covers} the edge $e$ if (i) $w$ covers $e$ and (ii) $t\in
\lambda (e)$, i.e.~the edge $e$ is \emph{active} during the time slot $t$;
in that case, edge $e$ is \emph{temporally covered} by the vertex appearance 
$(w,t)$. We now introduce the notion of a \emph{temporal vertex cover} and
the optimization problem \textsc{Temporal Vertex Cover}.

\begin{definition}
\label{temporal-vertex-cover-def}Let $(G,\lambda )$ be a temporal graph. A 
\emph{temporal vertex cover} of $(G,\lambda )$ is a temporal vertex subset $%
\mathcal{S}\subseteq\{(v,t):v\in V,1\leq t\leq T\}$ of $(G,\lambda )$ such
that every edge $e\in E$ is \emph{temporally covered} by at least one vertex
appearance $(w,t)$ in $\mathcal{S}$.
\end{definition}

\vspace{0,1cm} \noindent \fbox{ 
\begin{minipage}{0.96\textwidth}
 \begin{tabular*}{\textwidth}{@{\extracolsep{\fill}}lr} \textsc{Temporal Vertex Cover} \ \ (\textsc{TVC}) & \\ \end{tabular*}
 
  \vspace{1.2mm}
{\bf{Input:}}  A temporal graph $(G,\lambda)$.\\
{\bf{Output:}} A temporal vertex cover $\mathcal{S}$ of $(G,\lambda)$ with the smallest cardinality $|\mathcal{S}|$.
\end{minipage}} \vspace{0,3cm}

Note that \emph{\textsc{TVC}} is a natural temporal extension of the problem 
\textsc{Vertex Cover} on static graphs. 
In fact, \textsc{Vertex Cover} is the special case of \textsc{TVC} where $T=1$. 
Thus \textsc{TVC} is clearly NP-complete, as it also trivially belongs to NP.

\subsection{\label{Sliding-TVC-subsec}Sliding Window Temporal Vertex Cover}

In many real-world applications where scalability is important, the
lifetime $T$ can be arbitrarily large. In
such cases it may not be satisfactory enough that an edge is temporally
covered just \emph{once} during the whole lifetime of the temporal graph.
Instead, in such cases it makes sense that every edge is temporally covered
by some vertex appearance at least once during \emph{every small period} $%
\Delta$ of time, regardless of how large the lifetime $T$ is. Motivated by
this, we introduce in this section a natural \emph{sliding window} variant
of the \emph{\textsc{TVC}} problem, which offers a greater scalability of
the solution concept.

For every time slot $t\in \lbrack 1,T-\Delta +1]$, we define the \emph{%
time window} $W_{t}=[t,t+\Delta -1]$ as the sequence of the $\Delta $
consecutive time slots $t,t+1,\ldots ,t+\Delta -1$. We denote by $\mathcal{W}%
(T,\Delta )=\{W_{1},W_{2},\ldots ,W_{T-\Delta +1}\}$ the set of all time
windows in the lifetime of $(G,\lambda )$. Furthermore we denote by $%
E[W_{t}]=\bigcup _{i\in W_{t}}E_{i}$ the union of all edges appearing at least
once in the time window $W_{t}$. Finally we denote by $\mathcal{S}%
[W_{t}]=\{(v,t)\in \mathcal{S}:t\in W_{t}\}$ the restriction of the temporal
vertex subset $\mathcal{S}$ to the window $W_{t}$. We are now ready to
introduce the notion of a \emph{sliding $\Delta $-window temporal vertex
cover} and the optimization problem \textsc{Sliding Window Temporal Vertex
Cover}.

\begin{definition}
\label{sliding-temporal-vertex-cover-def} Let $(G,\lambda)$ be a temporal
graph with lifetime $T$ and let $\Delta \leq T$. A \emph{sliding }$\Delta $%
\emph{-window temporal vertex cover} of $(G,\lambda )$ is a temporal vertex
subset $\mathcal{S}\subseteq\{(v,t):v\in V,1\leq t\leq T\}$ of $(G,\lambda )$
such that, for every time window $W_{t}$ and for every edge $e\in E[W_{t}]$, 
$e$ is \emph{temporally covered} by at least one vertex appearance $(w,t)$
in $\mathcal{S}[W_{t}]$.
\end{definition}

\vspace{0cm} \noindent \fbox{ 
\begin{minipage}{0.96\textwidth}
 \begin{tabular*}{\textwidth}{@{\extracolsep{\fill}}lr} \textsc{Sliding Window Temporal Vertex Cover} \ \ (\textsc{SW-TVC}) & \\ \end{tabular*}
 
  \vspace{1.2mm}
{\bf{Input:}}  A temporal graph $(G,\lambda)$ with lifetime $T$, and an integer $\Delta \leq T$.\\
{\bf{Output:}} A sliding $\Delta$-window temporal vertex cover $\mathcal{S}$ of $(G,\lambda)$ with the smallest cardinality $|\mathcal{S}|$.
\end{minipage}} \vspace{0,3cm}

Whenever the parameter $\Delta $ is a fixed constant, 
we will refer to the above problem as the $\Delta $-\textsc{TVC} 
(i.e.~$\Delta $ is now a part of the problem name). 
Note that the problem \textsc{TVC} defined in Section~\ref{TVC-subsec} is the 
special case of \textsc{SW-TVC} where $\Delta =T$, 
i.e.~where there is only one $\Delta$-window in the whole temporal graph. 
Another special case\footnote{The 
problem $1$-\textsc{TVC} has already been investigated under the name ``evolving vertex cover'' 
in the context of maintenance algorithms in dynamic graphs~\cite{casteigts}; 
similar ``evolving'' variations of other graph covering problems have also been considered, e.g.~the ``evolving dominating set''~\cite{flocchini2}.} 
of \textsc{SW-TVC} is the problem $1$-\textsc{TVC}, 
whose optimum solution is obtained by iteratively solving the (static) 
problem \textsc{Vertex Cover} on each of the $T$ static instances of $(G,\lambda)$; 
thus $1$-\textsc{TVC} fails to fully capture the time dimension in temporal graphs.

\section{Hardness and approximability of \textsc{TVC}}
\label{sec:hardness-approx-TVC}

In this section we investigate the complexity of \textsc{Temporal Vertex Cover} (\textsc{TVC}). 
In Theorems~\ref{thm:set-cover-star} and~\ref{thm:hitting-set-star} we prove our 
hardness results for \textsc{TVC} on star temporal graphs 
(i.e.~when the underlying graph~$G$ is a star, cf.~Definition~\ref{temporal-classes-def}), 
which is in wide contrast to the (trivial) solution of \textsc{Vertex Cover} on a static star graph. 
The hardness results are obtained via reductions to the problems \textsc{Set Cover} and \textsc{Hitting Set}, 
respectively. 
On the positive side we prove in Theorem~\ref{TVC-general-approximation} that, on general temporal graphs, 
\textsc{TVC} can be approximated within a factor of $H_{n-1} - \frac{1}{2}\approx \ln n$, 
via a reduction to \textsc{Set Cover}.

\begin{theorem}
\label{thm:set-cover-star}
\textsc{TVC} on star temporal graphs is NP-complete. 
Furthermore, for any $\varepsilon>0$, \textsc{TVC} on star temporal graphs 
does not admit any polynomial-time $(1-\varepsilon)\ln{n}$-approximation algorithm, 
unless NP has $n^{O(\log \log n)}$-time deterministic algorithms.
\end{theorem}

\begin{proof}
First we reduce \textsc{Set Cover} to \textsc{TVC} on star temporal graphs.
 
\vspace{0,1cm} \noindent \fbox{ 
\begin{minipage}{0.96\textwidth}
 \begin{tabular*}{\textwidth}{@{\extracolsep{\fill}}lr} \textsc{Set Cover} & \\ \end{tabular*}
 
  \vspace{1.2mm}
{\bf{Input:}}  A universe $U = \{1,2,\ldots, n\}$ and a collection of $\mathcal{C}=\{C_1,C_2,\ldots,C_m\}$ of $m$ subsets of $U$ such that $\bigcup_{i=1}^{m}C_{i}=U$.\\
{\bf{Output:}} A subset $\mathcal{C}'\subseteq\mathcal{C}$ with the smallest cardinality such that $\bigcup_{C_i \in \mathcal{C}'} C_i =U$.
\end{minipage}} \vspace{0,3cm}
 
Given an instance $(U,\mathcal{C})$ of \textsc{Set Cover}, we construct an equivalent instance $(G,\lambda)$ 
of \textsc{TVC} on star temporal graphs as follows. 
We set $T = m$ and we let $G$ be a star graph on $n+1$ vertices, with center $c$ 
and leaves $v_1, v_2, \ldots, v_n$. 
The labeling $\lambda$ is such that, at every time slot $i\in[1,m]$, $G_i$ contains 
a set of isolated vertices and 
a star centered at $c$ and having the vertices $v_j$ as leaves, where $j \in C_i$.

We will now prove that there exists a temporal vertex cover $\mathcal{S}$ in $(G,\lambda)$ 
such that $|\mathcal{S}|\leq k$ if and only if there exists a set cover $\mathcal{C}'$ of $(U,\mathcal{C})$ such that $|\mathcal{C}'|\leq k$.

\begin{itemize}
\item[($\Rightarrow$)] Let $\mathcal{S}$ be a minimum-cardinality temporal vertex cover of $(G,\lambda)$, 
and let $|\mathcal{S}|\leq k$. 
Since $\mathcal{S}$ has minimum cardinality and $G$ is a star, we can assume without loss of generality that 
for every $i=1,2,\ldots, m$, 
either $\mathcal{S}_{i} = \{(c,i)\}$ or $\mathcal{S}_{i}=\emptyset$. 
Then the collection $\mathcal{C}' = \{C_{i}\in\mathcal{C} : S_i \not= \emptyset \}$ is a set cover of $U$.
Indeed, if $S_i \not= \emptyset$ then the appearance $(c,i)$ in $\mathcal{S}$ covers all the 
edges $cv_j$ of $G$, where $j\in C_i$. Thus, as the sequence of all non-empty sets $S_i$
covers all edges of $(G,\lambda)$, it follows that the union of all sets $C_i\in \mathcal{C}'$ 
covers all elements of $U=\{1,2,\ldots,n\}$. 
Finally, since $|\mathcal{S}|\leq k$, it follows by construction that $|\mathcal{C}|\leq k$ .

\item[($\Leftarrow$)] Let $\mathcal{C}^{\prime }$ be an optimal solution to \textsc{Set Cover}, 
and let $|\mathcal{C}'|\leq k$. 
We define the temporal vertex set $\mathcal{S} = \{(c,i) : C_{i}\in \mathcal{C}'\}$. 
For every $C_{i}\in \mathcal{C}'$, the appearance of $(c,i)$ in $\mathcal{S}$ covers all edges $cv_{j}$ of $G$, where $j\in C_{i}$. Thus, as the sets of $\mathcal{C}'$ cover all elements of $U=\{1,2,\ldots,n\}$, 
it follows that $\mathcal{S}$ is a temporal vertex cover of $(G,\lambda)$.
Finally, since $|\mathcal{C}|\leq k$, it follows by construction that $|\mathcal{S}|\leq k$ .
\end{itemize}

Therefore \textsc{TVC} on star temporal graphs is NP-complete. 
Moreover it is known that, for any $\varepsilon>0$, \textsc{Set Cover} cannot be approximated 
in polynomial time within a factor of $(1-\varepsilon)\ln n$ unless NP has $n^{O(\log \log n)}$-time 
deterministic algorithms~\cite{Feige98}. 
Therefore, due to the above polynomial-time reduction, it follows that \textsc{TVC} on star temporal graphs does not admit such an approximation algorithm as well.
\end{proof}

In the next theorem we complement our hardness results for \textsc{TVC} on star temporal graphs by 
reducing \textsc{Hitting Set} to it.

\begin{theorem}
\label{thm:hitting-set-star}
For every $\varepsilon<1$, \textsc{TVC} on star temporal graphs cannot be optimally solved 
in~$O(2^{\varepsilon T})$ time, unless the Strong Exponential Time Hypothesis (SETH) fails.
\end{theorem}
\begin{proof}
The proof is done via a reduction of \textsc{Hitting Set} to \textsc{TVC} on star temporal graphs. 
This reduction has a similar flavor as the one presented in Theorem~\ref{thm:set-cover-star}, 
since the problems \textsc{Set Cover} and \textsc{Hitting Set} are in a sense dual to each other
\footnote{That is seen by observing that an instance of \textsc{Set Cover}, with $U=\{1,\ldots,n\}$ and $\mathcal{C}=\{C_1, \ldots, C_m\}$, can be viewed as an instance of \textsc{Hitting Set}, with $U^{\ast}=\{1,\ldots,m\}$ and $\mathcal{C}^{\ast}=\{C_1^{\ast}, \ldots, C_n^{\ast} \}$, where $C_i^{\ast} = \{j : i \in C_j \}$, and vice versa.}. 
We first present the definition \textsc{Hitting Set}.

\vspace{0,1cm} \noindent \fbox{ 
\begin{minipage}{0.96\textwidth}
 \begin{tabular*}{\textwidth}{@{\extracolsep{\fill}}lr} \textsc{Hitting Set} & \\ \end{tabular*}
 
  \vspace{1.2mm}
{\bf{Input:}}  A universe $U = \{1,2,\ldots, n\}$ and a collection of $\mathcal{C}=\{C_1,C_2,\ldots,C_m\}$ of $m$ subsets of $U$ such that $\bigcup_{i=1}^{m}C_{i}=U$.\\
{\bf{Output:}} A subset $U'\subseteq U$ with the smallest cardinality such that $U'$ contains at least one element from each set in $\mathcal{C}$.
\end{minipage}} \vspace{0,3cm}

Given an instance $(U,\mathcal{C})$ of \textsc{Hitting Set}, we construct an equivalent instance 
of \textsc{TVC} on star temporal graphs as follows. 
We set $T = n$ and let $G$ be a star on $m+1$ vertices with center vertex $c$ 
and leaves $v_1,v_2, \ldots, v_m$. 
The labeling $\lambda$ is such that, at every time slot $i\in [1,n]$, 
$G_i$ contains a set of isolated vertices and 
a star centered at $c$ and having the vertices $v_j$ as leaves, where $i \in C_j$. 
That is, the $m$ leaves now correspond to the $m$ subsets in the family $\mathcal{C}$. 
Following a similar argumentation as in Theorem~\ref{thm:set-cover-star}, it follows that 
there exists a temporal vertex cover $\mathcal{S}$ in $(G,\lambda)$ 
such that $|\mathcal{S}|\leq k$ if and only if there exists a hitting set $U'\subseteq U$ 
of $(U,\mathcal{C})$ such that $|U'|\leq k$.

Assume now that there exists an $O(2^{\varepsilon T})$-time algorithm for optimally solving \textsc{TVC} on star temporal graphs, for some $\varepsilon<1$. 
Then, due to the above reduction from \textsc{Hitting Set}, we can use this algorithm to optimally solve 
\textsc{Hitting Set} in $O(2^{\varepsilon n})$-time. 
This is a contradiction, unless the Strong Exponential Time Hypothesis (SETH) fails~\cite{Cygan16}. 
\end{proof}

Note that the above construction in the proof of Theorem \ref{thm:hitting-set-star} 
can be trivially reverted, thus providing the inverse reduction, i.e.~from
\textsc{TVC} on star temporal graphs to \textsc{Hitting Set}. 
In the obtained instance of \textsc{Hitting Set} produced by this reduction, 
the maximum size of any set is equal to the maximum number $\ell$ of time slots 
at which any fixed edge appears. 
Therefore, whenever $k,\ell$ are constants, 
known results on \textsc{Hitting Set} (see e.g.~\cite{niedermeier2003efficient}) immediately imply 
a linear-time algorithm for deciding whether there is a temporal vertex cover of size at most $k$ 
in a star temporal graph, in which any fixed edge appears in at most $\ell$ time slots.
Now we provide our positive results for \textsc{TVC} on general temporal graphs.

\begin{theorem} \label{TVC-general-approximation}
\textsc{TVC} on general temporal graphs can be approximated in polynomial time within a factor of at most $H_{n-1} - \frac{1}{2}$.
\end{theorem}

\begin{proof}
The proof is done via a reduction of \textsc{TVC} to \textsc{Set Cover}. 
Given an instance $(G,\lambda)$ of \textsc{TVC}, we construct an instance $(U,\mathcal{C})$ of \textsc{Set Cover} as follows. 
We set the universe $U$ to be $E(G)$, and for every vertex appearance $(v,i)$ of $(G,\lambda)$ we add to $\mathcal{C}$ the set $C_{v,i}$
of those edges that are temporally covered by $(v,i)$, i.e. $C_{v,i} = \{e:~v\text{ is an endpoint of } e \text{ and } e\in E_i\}$. 
Note that, in this instance of \textsc{Set Cover}, every set $C_{v,i}$ has at most $n-1$ elements.

Now we show that there exists a temporal vertex cover $\mathcal{S}$ in $(G,\lambda)$ 
with $|\mathcal{S}|\leq k$ if and only if there exists a set cover $\mathcal{C}' \subseteq \mathcal{C}$ of $U$ with $|\mathcal{C}'|\leq k$.

\begin{itemize}
	\item[($\Rightarrow$)] Let $\mathcal{S}$ be a temporal vertex cover of $(G,\lambda)$ 
	with $|\mathcal{S}|\leq k$. 
	Then there exist $k$ vertex appearances that temporally cover all edges of $G$. By the construction, the sets in $\mathcal{C}$,
	corresponding to these vertex appearances, cover $U$ and therefore they constitute a set cover $\mathcal{C}'$ of size at most $k$.
	
	\item[($\Leftarrow$)] Let $\mathcal{C}^{\prime }$ be a set cover of $(U,\mathcal{C})$ with $|\mathcal{C}'|\leq k$.
	Since every $C_{v,i} \in \mathcal{C}^{\prime}$ contains the edges that are temporally covered by $(v,i)$,
	it follows that $\mathcal{S} = \{(v,i) : C_{v,i}\in \mathcal{C}'\}$ is a temporal vertex cover of $(G,\lambda)$ of size at most $k$.
\end{itemize}

Therefore we can compute an approximate solution to \textsc{TVC} on $(G,\lambda)$ by first 
computing an approximate solution to \textsc{Set Cover} on $(U,\mathcal{C})$, 
achieving the same approximation factor for both problems. 
Now we apply the polynomial-time approximation algorithm of~\cite{DuhFurer97} for \textsc{Set Cover} 
which achieves a ratio of $H_{n-1}-\frac{1}{2}$, where $H_{n-1}$ is the $(n-1)$-th harmonic number 
and $n-1$ is an upper bound on the size of the sets in the instance $(U,\mathcal{C})$. 
Since $H_{n-1} \leq \ln n + 1$, the theorem follows.
\end{proof}

\section{An almost tight algorithm for \textsc{SW-TVC}}
\label{sec:tight-algorithms-SW-TVC}

In this section we investigate the complexity of \textsc{Sliding Window Temporal Vertex Cover} (\textsc{SW-TVC}). 
First we prove in Section~\ref{subsec:lower-bound-ETH} a strong lower bound on the complexity of 
optimally solving this problem on arbitrary temporal graphs. 
More specifically we prove that, for \emph{any} (arbitrarily growing) 
functions $f:\mathbb{N}\rightarrow \mathbb{N}$ and $g:\mathbb{N}\rightarrow \mathbb{N}$, there 
exists a constant $\varepsilon \in (0,1)$ such that \textsc{SW-TVC} cannot be solved 
in $f(T)\cdot 2^{\varepsilon n \cdot g(\Delta)}$ time, assuming the Exponential Time Hypothesis (ETH). 
This ETH-based lower bound turns out to be asymptotically almost tight. 
In fact, we present in Section~\ref{subsec:dynam-prog} an exact dynamic programming algorithm for \textsc{SW-TVC} whose running time on an arbitrary temporal graph is $O(T\Delta (n+m)\cdot 2^{n(\Delta+1)})$, 
which is asymptotically almost optimal, assuming~ETH. 
That is, although the running time of our algorithm has an exponential part $2^{n (\Delta+1)}$, 
there does not exist (assuming ETH) any algorithm whose running time is subexponential in $n$ 
(i.e.~having $o(n)$ instead of $n$ in the exponent), 
even if we allow an \emph{arbitrarily growing} function~$g(\Delta)$ of $\Delta$ in the exponent. 
In Section~\ref{sec:boundedVCnum} we prove that our algorithm can be refined so that, 
when the vertex cover number of each snapshot $G_{i}$ is bounded by a constant~$k$, 
the running time becomes $O(T\Delta (n+m)\cdot n^{k(\Delta+1)})$. 
That is, whenever $\Delta$ is constant, the algorithm is polynomial in the input size on temporal graphs with bounded vertex cover at every time slot. 
Notably, for the class of always star temporal graphs (i.e.~a superclass of the star temporal graphs 
studied in Section~\ref{sec:hardness-approx-TVC}) the running time of the algorithm is $O(T\Delta (n+m)\cdot 2^{\Delta})$.

\subsection{A complexity lower bound}
\label{subsec:lower-bound-ETH}

In the classic textbook NP-hardness reduction from \textsc{3SAT} to \textsc{Vertex Cover} 
(see e.g.~\cite{gareyjohnson}), the produced instance of \textsc{Vertex Cover} is a graph whose number 
of vertices is linear in the number of variables and clauses of the \textsc{3SAT} instance. 
Therefore the next theorem follows by Theorem~\ref{ETH-thm} (for a discussion 
see also~\cite{Lokshtanov11}). 
\begin{theorem}
\label{Vetrex-Cover-lower-bound}
Assuming ETH, there exists an $\varepsilon_{0}<1$ such that \textsc{Vertex Cover} 
cannot be solved in $O(2^{\varepsilon_{0}n})$, where $n$ is the number of vertices.
\end{theorem}
In the the following theorem we prove a strong ETH-based lower bound for \textsc{SW-TVC}. 
This lower bound is asymptotically almost tight, as we present in Section~\ref{subsec:dynam-prog} 
a dynamic programming algorithm for \textsc{SW-TVC} with running time $O(T\Delta(n+m)\cdot 2^{n\Delta})$, 
where $n$ and~$m$ are the numbers of vertices and edges in the underlying graph $G$, respectively.

\begin{theorem}
\label{SW-TVC-lower-bound}
For \emph{any} two (arbitrarily growing) functions $f:\mathbb{N}\rightarrow \mathbb{N}$ 
and $g:\mathbb{N}\rightarrow \mathbb{N}$, there exists a constant $\varepsilon \in (0,1)$ 
such that \textsc{SW-TVC} cannot be solved in $f(T)\cdot 2^{\varepsilon n \cdot g(\Delta)}$ time 
assuming ETH, where $n$ is the number of vertices in the underlying graph $G$ of the temporal graph.
\end{theorem}

\begin{proof}
To prove the theorem, we reduce \textsc{Vertex Cover} to a restricted version of \textsc{SW-TVC}. 
In our reduction, we construct an instance of \textsc{SW-TVC} which forces us to 
solve \textsc{Vertex Cover} on a particular static instance. 
Let $f:\mathbb{N}\rightarrow \mathbb{N}$ and $g:\mathbb{N}\rightarrow \mathbb{N}$ be two functions. 
Assume, for the sake of contradiction, that for every $\varepsilon<1$ 
there exists an algorithm~$\mathcal{A}_{\varepsilon}$ which solves \textsc{SW-TVC} 
in $f(T)\cdot 2^{\varepsilon n \cdot g(\Delta)}$ time. 
Now fix arbitrary numbers $T_0,\Delta_0\in\mathbb{N}$, where $\Delta_0 \leq T_0$. 
We reduce \textsc{Vertex Cover} to the restriction of \textsc{SW-TVC} where $T=T_0$ and $\Delta=\Delta_0$ 
in the input, as follows. 
Let $G$ be an instance graph of \textsc{Vertex Cover} with $n$ vertices. We construct the temporal graph $(G,\lambda)$ 
with snapshots $G_1, G_2, \ldots, G_{T_0}$, 
such that $G_i = G$ whenever $i =1\mod \Delta_0$, and $G_i$ is an independent set on $n$ vertices otherwise. 
Then, every optimum solution of \textsc{SW-TVC} on~$(G,\lambda)$ contains an optimum solution 
of \textsc {Vertex Cover} on $G$ at each time slot~$i$, where $i =1\mod \Delta_0$.

By our assumption, we solve \textsc{SW-TVC} on the instance $(G,\lambda)$ 
using algorithm~$\mathcal{A}_{\varepsilon}$, where $\varepsilon = \frac{\varepsilon_{0}}{g(\Delta_0)}$ 
and $\varepsilon_{0}$ is the constant of Theorem~\ref{Vetrex-Cover-lower-bound} for \textsc{Vertex Cover}. 
Note that $\varepsilon$ is a constant, since both~$\varepsilon_{0}$ and~$g(\Delta_0)$ are constants. 
Furthermore note that the result of the algorithm also provides a minimum vertex cover in the original (static) graph $G$. 
The running time of $\mathcal{A}_{\varepsilon}$ is 
by assumption~$f(T_0)\cdot 2^{\frac{\varepsilon_{0}}{g(\Delta_0)} n \cdot g(\Delta_0)} = f(T_0)\cdot 2^{\varepsilon_{0} n}$.
Therefore, since $f(T_0)$ is also a constant, 
the existence of the algorithm $\mathcal{A}_{\varepsilon}$ for \textsc{SW-TVC} implies an algorithm 
for \textsc{Vertex Cover} with running time $O(2^{\varepsilon_{0}n})$, which is a contradiction, assuming ETH, due to Theorem~\ref{Vetrex-Cover-lower-bound}.
\end{proof}

\subsection{An exact dynamic programming algorithm}
\label{subsec:dynam-prog}

The main idea of our dynamic programming algorithm for \textsc{SW-TVC} is to scan the temporal graph from left to right 
with respect to time (i.e.~to scan the snapshots $G_{i}$ increasingly on $i$), 
and at every time slot to consider all possibilities for the vertex appearances at the previous $\Delta$ time slots. 
Before we proceed with the presentation and analysis of our algorithm, 
we start with a simple but useful observation.

\begin{observation}
\label{obs:feasible-partial} Let $(G,\lambda )$ be a temporal graph with
lifetime $T$. Let $\mathcal{S}$ be a sliding $\Delta $-window temporal
vertex cover of $(G,\lambda )$. Then, for every $\Delta \leq t\leq T$, the
temporal vertex subset~$\mathcal{S}|_{[1,t]}=\{(v,i)\in \mathcal{S}:i\leq
t\} $ is a sliding $\Delta $-window temporal vertex cover of~$(G,\lambda )|_{[1,t]}$.
\end{observation}
\begin{proof}[Proof sketch.]
	By definition, $\mathcal{S}$ is such that for \emph{every} time window $W_j,~j=1,\ldots, T-\Delta+1$, and every $e \in E[W_j]$, $e$ is temporally covered by some vertex appearance in $\mathcal{S}[W_j]$. 
	Restrict $(G,\lambda )$ and the vertex appearances of $\mathcal{S}$ to the time slots $1,\ldots ,t$. For every time window $W_j,~j=1,\ldots, t-\Delta+1$, and every $e \in E[W_j]$, $e$ is still temporally covered by some vertex appearance in $\mathcal{S}|_{[1,t]}[W_j]$, since the time windows $W_j,~j=1,\ldots, t-\Delta+1$, are the same in both $(G,\lambda )$ and $(G,\lambda )|_{[1,t]}$, hence we only need vertex appearances in $\mathcal{S}|_{[1,t]}$ to temporally cover edges in $(G,\lambda )|_{[1,t]}$.
\end{proof}

Let $(G,\lambda )$ be a temporal graph with $n$ vertices and lifetime $T$,
and let $\Delta \leq T$. For every $t=1,2,\ldots, T-\Delta +1$ and every $%
\Delta $-tuple of vertex subsets $A_{1},\ldots A_{\Delta }$ of $G$, we
define $f(t;A_{1},A_{2},\ldots ,A_{\Delta })$  
to be the smallest cardinality of a sliding $\Delta $-window temporal vertex
cover $\mathcal{S}$ of~$(G,\lambda )|_{[1,t+\Delta -1]}$, such that $%
\mathcal{S}_{t}=(A_{1},t),\ \mathcal{S}_{t+1}=(A_{2},t+1),\ \ldots ,\ 
\mathcal{S}_{t+\Delta -1}=(A_{\Delta },t+\Delta -1)$. If there exists no
sliding $\Delta $-window temporal vertex cover~$\mathcal{S}$ of~$(G,\lambda
)|_{[1,t+\Delta -1]}$ with these prescribed vertex appearances in the time
slots $t,t+1,\ldots ,t+\Delta -1$, then we define $f(t;A_{1},A_{2},\ldots
,A_{\Delta })=\infty $. Note that, once we have computed all possible values
of the function~$f(\cdot )$, then the optimum
solution of \textsc{SW-TVC} on $(G,\lambda )$ has cardinality%
\begin{equation}
\text{OPT}_{\textsc{SW-TVC}}(G,\lambda) = 
\min_{A_{1},A_{2},\ldots,A_{\Delta }\subseteq V}
\left\{f(T-\Delta +1;A_{1},A_{2},\ldots ,A_{\Delta})\right\} .  \label{optimum-dynamic-eq}
\end{equation}

\begin{observation}
\label{obs:infeasibility-last window-dynamic}
If the temporal vertex set $\bigcup\nolimits_{i=1}^{\Delta }(A_{i},t+i-1)$ is not a temporal vertex
cover of $(G,\lambda )|_{[t,t+\Delta -1]}$ then $f(t;A_{1},A_{2},\ldots,A_{\Delta })=\infty $.
\end{observation}

Due to Observation~\ref{obs:infeasibility-last window-dynamic} we assume 
below without loss of generality that $\bigcup\nolimits_{i=1}^{\Delta}(A_{i},t+i-1)$ 
is a temporal vertex cover of $(G,\lambda)|_{[t,t+\Delta -1]}$. 
We are now ready to present our main recursion formula in the next lemma.

\begin{lemma}
\label{recursion-dynamic-lem}Let $(G,\lambda )$ be a temporal graph, where $%
G=(V,E)$. Let $2\leq t\leq T-\Delta +1$ and let $A_{1},A_{2},\ldots
A_{\Delta }$ be a $\Delta $-tuple of vertex subsets of the underlying graph $%
G$. Suppose that $\bigcup\nolimits_{i=1}^{\Delta }(A_{i},t+i-1)$ is a
temporal vertex cover of $(G,\lambda )|_{[t,t+\Delta -1]}$. Then%
\begin{equation}
f(t;A_{1},A_{2},\ldots ,A_{\Delta })=|A_{\Delta }|+\min_{X\subseteq
V}\left\{ f(t-1;X,A_{1},\ldots ,A_{\Delta -1})\right\} .
\label{recursion-eq}
\end{equation}
\end{lemma}

\begin{proof}
First consider the case where $\min_{X\subseteq V}\left\{ f(t-1;X,A_{1},\ldots ,A_{\Delta
-1})\right\} =\infty $. Assume that $f(t;A_{1},A_{2},\ldots ,A_{\Delta
})\neq \infty $ and let $\mathcal{S}$ be a sliding $\Delta $-window temporal
vertex cover of the instance $(G,\lambda )|_{[1,t+\Delta -1]}$, in which the
vertex appearances in the the last $\Delta $ time slots $t,t+1,\ldots
,t+\Delta -1$ are given by $\mathcal{S}_{t}=(A_{1},t),\ \mathcal{S}%
_{t+1}=(A_{2},t+1),\ \ldots ,\ \mathcal{S}_{t+\Delta -1}=(A_{\Delta
},t+\Delta -1)$. Then, by Observation~\ref{obs:feasible-partial}, $\mathcal{S%
}|_{[1,t+\Delta -2]}$ is a sliding $\Delta $-window temporal vertex cover of
the instance $(G,\lambda )|_{[1,t+\Delta -2]}$. Moreover, the vertex
appearances of $\mathcal{S}|_{[1,t+\Delta -2]}$ in the the last $\Delta -1$
time slots $t,t+1,\ldots ,t+\Delta -2$ are given by $\mathcal{S}%
_{t}=(A_{1},t),\ \mathcal{S}_{t+1}=(A_{2},t+1),\ \ldots ,\ \mathcal{S}%
_{t+\Delta -2}=(A_{\Delta -1},t+\Delta -2)$. 
Now let $X$ be the set of vertices of $G$ which are active in $\mathcal{S}|_{[1,t+\Delta -2]}$ 
during the time slot $t-1$. 
Then note that $f(t-1;X,A_{1},\ldots ,A_{\Delta-1})$ is upper-bounded by the cardinality of  $\mathcal{S}|_{[1,t+\Delta -2]}$, 
and thus $f(t-1;X,A_{1},\ldots,A_{\Delta -1})\neq \infty $, which is a contradiction. 
That is, if $\min_{X\subseteq V}\left\{ f(t-1;X,A_{1},\ldots ,A_{\Delta -1})\right\}
=\infty $ then also $f(t;A_{1},A_{2},\ldots ,A_{\Delta })=\infty $, and in
this case the value of $f(t;A_{1},A_{2},\ldots ,A_{\Delta })$ is correctly
computed by (\ref{recursion-eq}).

Now consider the case where $\min_{X\subseteq V}\left\{ f(t-1;X,A_{1},\ldots ,A_{\Delta
-1})\right\} \neq \infty $, and let $X\subseteq V$ be a vertex subset for
which $f(t-1;X,A_{1},\ldots ,A_{\Delta -1})$ is minimized. Furthermore let $%
\mathcal{S}$ be a minimum-cardinality sliding $\Delta $-window temporal
vertex cover of the instance $(G,\lambda )|_{[1,t+\Delta -2]}$, in which the
vertex appearances in the the last $\Delta $ time slots $t-1,t,\ldots
,t+\Delta -2$ are given by $\mathcal{S}_{t-1}=(X,t-1),\ \mathcal{S}%
_{t}=(A_{1},t),\ \ldots ,\ \mathcal{S}_{t+\Delta -2}=(A_{\Delta -1},t+\Delta
-2)$. Then $\mathcal{S}\cup (A_{\Delta },t+\Delta -1)$ is a sliding $\Delta $%
-window temporal vertex cover of the instance $(G,\lambda )|_{[1,t+\Delta
-1]}$, since $\bigcup\nolimits_{i=1}^{\Delta }(A_{i},t+i-1)$ is a temporal
vertex cover of $(G,\lambda )|_{[t,t+\Delta -1]}$ by the assumption of the
lemma. Thus%
\begin{eqnarray}
f(t;A_{1},A_{2},\ldots ,A_{\Delta }) &\leq &|\mathcal{S}\cup (A_{\Delta
},t+\Delta -1)| \notag \\
&=&|A_{\Delta }|+|\mathcal{S}|  \notag \\
&=&|A_{\Delta }|+\min_{X\subseteq V}\left\{ f(t-1;X,A_{1},\ldots ,A_{\Delta
-1})\right\} ,  \label{recursion-proof-eq-1}
\end{eqnarray}%
and thus, in particular, $f(t;A_{1},A_{2},\ldots ,A_{\Delta })\neq \infty $.
Now let $\mathcal{S}'$ be a minimum-cardinality sliding $\Delta$-window temporal vertex cover of the instance $(G,\lambda )|_{[1,t+\Delta-1]}$, in which the vertex appearances in the the last $\Delta $ time slots $t,t+1,\ldots ,t+\Delta -1$ are given 
by $\mathcal{S}'_{t}=(A_{1},t),\ \mathcal{S}'_{t+1}=(A_{2},t+1),\ \ldots ,\ \mathcal{S}'_{t+\Delta
-1}=(A_{\Delta },t+\Delta -1)$. Note that $|\mathcal{S}^{\prime }|=$ $%
f(t;A_{1},A_{2},\ldots ,A_{\Delta })$, since $\mathcal{S}^{\prime }$ has
minimum cardinality by assumption. Observation~\ref{obs:feasible-partial}
implies that the temporal vertex subset $\mathcal{S}^{\prime \prime }=%
\mathcal{S}^{\prime }|_{[1,t+\Delta -2]}$ is a sliding $\Delta $-window
temporal vertex cover\ of the instance $(G,\lambda )|_{[1,t+\Delta -2]}$,
and thus $|\mathcal{S}^{\prime \prime }|\geq \min_{X\subseteq V}\left\{
f(t-1;X,A_{1},\ldots ,A_{\Delta -1})\right\} $. Furthermore, since $|%
\mathcal{S}^{\prime }|=|A_{\Delta }|+|\mathcal{S}^{\prime \prime }|$ by
construction, it follows that%
\begin{equation}
f(t;A_{1},A_{2},\ldots ,A_{\Delta })=|A_{\Delta }|+|\mathcal{S}^{\prime
\prime }|\geq |A_{\Delta }|+\min_{X\subseteq V}\left\{ f(t-1;X,A_{1},\ldots
,A_{\Delta -1})\right\} .  \label{recursion-proof-eq-2}
\end{equation}%
Summarizing, equations (\ref{recursion-proof-eq-1}) and (\ref{recursion-proof-eq-2})\
imply that $f(t;A_{1},A_{2},\ldots ,A_{\Delta })=|A_{\Delta
}|+\min_{X\subseteq V}\left\{ f(t-1;X,A_{1},\ldots ,A_{\Delta -1})\right\} $%
, whenever $\min_{X\subseteq V}\left\{ f(t-1;X,A_{1},\ldots ,A_{\Delta
-1})\right\} \neq \infty $. This completes the proof of the lemma.
\end{proof}

Using the recursive computation of Lemma~\ref{recursion-dynamic-lem}, we are
now ready to present Algorithm~\ref{alg:dynam} for computing the value of an
optimal solution of \textsc{SW-TVC} on a given arbitrary temporal graph $%
(G,\lambda )$. Note that Algorithm~\ref{alg:dynam} can be easily modified 
such that it also computes the actual optimum solution of 
\textsc{SW-TVC} (instead of only its optimum cardinality). The proof of
correctness and running time analysis of Algorithm~\ref{alg:dynam} are given
in the next theorem.

\begin{algorithm}[h]
\caption{\textsc{SW-TVC}}  
\label{alg:dynam}  
\begin{algorithmic}[1]
\REQUIRE{A temporal graph $(G, \lambda)$ with lifetime $T$, where $G=(V,E)$, 
	and a natural $\Delta \leq T$.}
\ENSURE{The smallest cardinality of a \SWTVC\ in $(G, \lambda)$.}

\medskip

\FOR{$t=1$ to $T-\Delta+1$} \label{alg-synamic-line-3}
	\FOR{all $A_{1},A_{2},\ldots,A_{\Delta}\subseteq V$} \label{alg-synamic-line-4}
	\vspace{0,1cm}
		\IF{$\bigcup_{i=1}^{\Delta}(A_{i},t+i-1)$ is a temporal vertex cover of $(G,\lambda)|_{[t,t+\Delta-1]}$} \label{alg-synamic-line-5}
			\IF{$t=1$} \label{alg-synamic-line-5a}
				\STATE{$f(t;A_{1},A_{2},\ldots,A_{\Delta}) \leftarrow \sum_{i=1}^{\Delta}|A_{i}|$} \label{alg-synamic-line-5b}
			\ELSE
				\STATE{$f(t;A_{1},A_{2},\ldots,A_{\Delta}) \leftarrow |A_{\Delta}|+ min_{X\subseteq V}\left\{ f(t-1;X,A_{1},\ldots ,A_{\Delta -1})\right\}$} \label{alg-synamic-line-6}
			\ENDIF
		\ELSE \label{alg-synamic-line-7}
			\STATE{$f(t;A_{1},A_{2},\ldots,A_{\Delta}) \leftarrow \infty$} \label{alg-synamic-line-8}
		\ENDIF
	\ENDFOR
\ENDFOR

\medskip

\RETURN{$min_{A_{1},\ldots ,A_{\Delta}\subseteq V} \left\{ f(T-\Delta+1;A_{1},\ldots ,A_{\Delta})\right\}$} \label{alg-synamic-line-9}
\end{algorithmic}
\end{algorithm}

\begin{theorem}
\label{dynamic-algorithm-thm}Let $(G,\lambda )$ be a temporal graph, where $%
G=(V,E)$ has $n$ vertices and $m$ edges. Let $T$ be its lifetime 
and let $\Delta $ be the length of the sliding window.
Algorithm~\ref{alg:dynam} computes in~$O(T\Delta (n+m)\cdot 2^{n(\Delta+1) })$
time the value of an optimal solution of \textsc{SW-TVC} on $(G,\lambda )$.
\end{theorem}

\begin{proof}
In its main part (lines~\ref{alg-synamic-line-3}-\ref{alg-synamic-line-8}),
the algorithm iterates over all time slots $1\leq t\leq T-\Delta +1$ and
over all vertex subsets $A_{1},A_{2},\ldots A_{\Delta }\subseteq V$.
Whenever it detects a tuple $(t;A_{1},A_{2},\ldots A_{\Delta })$ such that $%
\bigcup\nolimits_{i=1}^{\Delta }(A_{i},t+i-1)$ is not a temporal vertex
cover of $(G,\lambda )|_{[t,t+\Delta -1]}$, then it sets $%
f(t;A_{1},A_{2},\ldots A_{\Delta })=\infty $ in line~\ref{alg-synamic-line-8}. 
This is correct by Observation~\ref{obs:infeasibility-last window-dynamic}.

For all other tuples $(t;A_{1},A_{2},\ldots A_{\Delta })$, the algorithm
distinguishes in lines~\ref{alg-synamic-line-5a}-\ref{alg-synamic-line-6} the
cases $t=1$ and $t\geq 2$. If $t\geq 2$ the algorithm recursively computes
in line~\ref{alg-synamic-line-6} the value of $f(t;A_{1},A_{2},\ldots
A_{\Delta })$ using values that have been previously computed. The correctness of this recursive computation follows by Lemma~\ref{recursion-dynamic-lem}. 
If $t=1$, then clearly the optimum solution of \textsc{SW-TVC} on $(G,\lambda
)|_{[1,\Delta ]}$ has value equal to the total number of vertex appearances
in the time slots $1,2,\ldots ,\Delta $, i.e.~$f(1;A_{1},A_{2},\ldots
A_{\Delta })=\sum_{i=1}^{\Delta }|A_{i}|$, as it is computed in line~\ref{alg-synamic-line-5b}. 
Finally, the algorithm correctly returns in line~\ref{alg-synamic-line-9} 
the smallest value of $f(T-\Delta +1;A_{1},A_{2},\ldots A_{\Delta })$
among all possible $\Delta $-tuples $A_{1},A_{2},\ldots A_{\Delta }$. The
correctness of this step has been discussed above, in equation~(\ref{optimum-dynamic-eq}).

With respect to running time, Algorithm~\ref{alg:dynam} iterates for each
value $t=1,2,\ldots, T$ and for each of the $2^{n\Delta }$ different $\Delta$-tuples 
$A_{1},A_{2},\ldots A_{\Delta }\subseteq V$ in lines~\ref{alg-synamic-line-3}-\ref{alg-synamic-line-8}. 
The only non-trivial computations within these lines are in lines~\ref{alg-synamic-line-5} 
and~\ref{alg-synamic-line-6}. 
In line~\ref{alg-synamic-line-5} the algorithm checks whether $\bigcup\nolimits_{i=1}^{\Delta }(A_{i},t+i-1)$
is a temporal vertex cover of $(G,\lambda )|_{[t,t+\Delta -1]}$. This can be
done in~$O(\Delta (n+m))$ time, where $m$ is the number of edges in the
(static) underlying graph $G$, by simply enumerating all edges that are covered by 
the vertex appearances in $\bigcup\nolimits_{i=1}^{\Delta }(A_{i},t+i-1)$ 
and comparing their number with the total number of edges that are active at least once in 
the time window $W_t=[t,t+\Delta -1]$. 
On the other hand, to execute line~\ref{alg-synamic-line-6} we need at most~$O(2^n)$ time for computing the minimum among at most $2^n$ different known values. 
Similarly, to execute line~\ref{alg-synamic-line-9} we need at most $O(2^{n\Delta})$ time for computing the minimum among at most~$2^{n\Delta}$ different known values. 
Therefore the total running time of Algorithm~\ref{alg:dynam} 
is upper-bounded by $O(T\Delta (n+m)\cdot 2^{n(\Delta+1) })$ time.
\end{proof}

\subsection{Always bounded vertex cover temporal graphs}
\label{sec:boundedVCnum}

Let $(G,\lambda)$ be a temporal graph of lifetime $T$, and let $\mathcal{S}$ be a minimum-cardinality \SWTVC\ of $(G,\lambda)$. 
Note that the minimality of $|\mathcal{S}|$ implies that, for every $i=1,2,\ldots,T$, 
$|\mathcal{S}_i|$ is upper-bounded by the (static) vertex cover number of $G_i$. 
Therefore, in the recursive relation of Lemma~\ref{recursion-dynamic-lem}, 
it is enough to only consider subsets $X, A_1, A_2, \ldots, A_{\Delta}\subseteq V$ 
which have cardinality upper bounded by the vertex cover number in the corresponding snapshot. 
Thus, whenever $\Delta$ is constant, Algorithm~\ref{alg:dynam} can be modified to become 
a polynomial-time algorithm for the class of always bounded vertex cover temporal graphs. 
Formally, let $k$ be a constant and let $\mathcal{C}_k$ be the class of graphs having vertex cover number at most $k$. The next theorem follows now from the analysis of Theorem~\ref{dynamic-algorithm-thm}.

\begin{theorem}
\label{bounded-vc-number-thm}
\textsc{SW-TVC} on always $\mathcal{C}_k$ temporal graphs can be solved
in $O(T\Delta (n+m)\cdot n^{k(\Delta+1)})$ time.
\end{theorem}

In particular, in the special, yet interesting, case of always star temporal graphs 
(i.e.~where every snapshot $G_{i}$ is a star graph), 
our search at every step reduces to just one binary choice for each of the previous $\Delta$ time slots, 
of whether to include the central vertex of the star in a snapshot or not.
Hence we have the following theorem as a direct implication of Theorem~\ref{bounded-vc-number-thm}.

\begin{theorem}
\label{always-star-thm}
\textsc{SW-TVC} on always star temporal graphs can be solved in $O(T\Delta (n+m)\cdot 2^{\Delta})$ time.
\end{theorem}

Note here that, although for every time step $i=1,2,\ldots,T$ the size of a (static) minimum vertex cover of $G_i$ provides an upper bound on the number of vertex appearances that should be selected by a minimum-cardinality \SWTVC\ at time $i$, the \emph{set of vertices} selected at time $i$ by 
an optimum solution to \textsc{SW-TVC} need not be a subset of any (static) minimum vertex cover of $G_i$. 
We illustrate this with the example of Figure~\ref{example_static_VC-fig}.

\begin{figure}[h]
	\centering
	\includegraphics[width=0.48\textwidth]{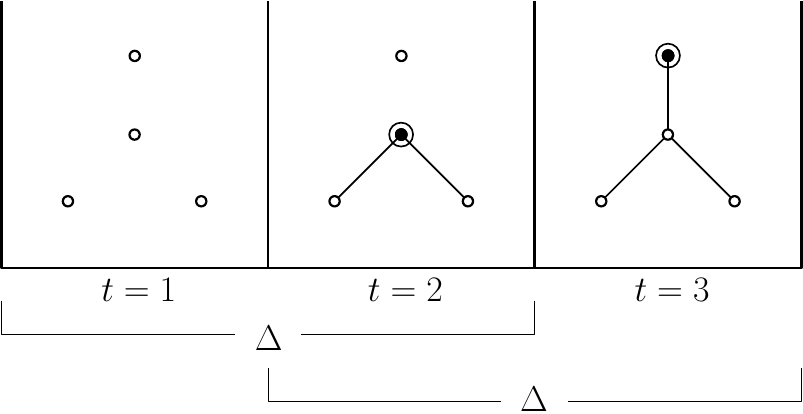}  
	\caption{A minimum-cardinality \SWTVC\ which selects a vertex appearance at time step $t=3$, whose corresponding vertex belongs in no minimum vertex cover of $G_3$.}
	\label{example_static_VC-fig}
\end{figure}

\section{Approximation hardness of \textsc{2-TVC}}
\label{sec:hardDTVC}

In this section we study the complexity of \DTVCproblem\, where $\Delta$ is constant.
We start with an intuitive observation that, for every fixed $\Delta$, the problem \DTVCproblemD{($\Delta+1$)} is at least as hard as \DTVCproblem.
Indeed, let $\mathcal{A}$ be an algorithm that computes a minimum-cardinality \SWTVCD{$(\Delta+1)$} of $(G, \lambda)$. 
It is easy to see that a minimum-cardinality \SWTVC\ of $(G, \lambda)$ can also be computed using $\mathcal{A}$, 
if we amend the input temporal graph by inserting one edgeless snapshot after every 
$\Delta$ consecutive snapshots of $(G,\lambda)$, see~Figure~\ref{fig:Delta-increase}.

\begin{figure}[h]
	\centering
	\includegraphics[scale=0.8]{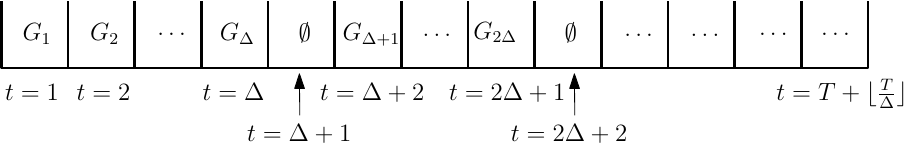}  
	\caption{Inserting ``empty'' time slots to compute a minimum-cardinality \SWTVC\ on $(G, \lambda)$ 
	using algorithm $\mathcal{A}$ for \DTVCproblemD{($\Delta+1$)}.}
	\label{fig:Delta-increase}
\end{figure}

Since the \DTVCproblemD{$1$} problem is equivalent to solving $T$ instances of \textsc{Vertex Cover}
(on static graphs), the above reduction demonstrates in particular 
that, for any natural $\Delta$, \DTVCproblem\ is at least as hard as \textsc{Vertex Cover}.
Therefore, if \textsc{Vertex Cover} is hard for a class $\mathcal{X}$ of static graphs, 
then \DTVCproblem\ is also hard for the class of always $\mathcal{X}$ temporal graphs. 
In this section, we show that the converse is not true.
Namely, we reveal a class $\mathcal{X}$ of graphs, for which \textsc{Vertex Cover} can be solved in \emph{linear} time, 
but \DTVCproblemD{2} is NP-hard on always $\mathcal{X}$ temporal graphs.
In fact, we show the even stronger result that \DTVCproblemD{2} does not admit a PTAS on always $\mathcal{X}$ temporal graphs, unless P=NP.

To prove the main result (in Theorem~\ref{th:2DeltaSHard}) we start with an auxiliary lemma, showing that, unless P=NP, \textsc{Vertex Cover} 
does not admit a PTAS on the class $\mathcal{Z}$ of graphs which can be obtained from a cubic graph by subdividing every edge exactly 4 times.

\begin{lemma}\label{lem:apxY}
\textsc{Vertex Cover} on $\mathcal{Z}$ does not admit a PTAS, unless P=NP.
\end{lemma}
\begin{proof}
	\textsc{Vertex Cover} does not admit a PTAS when the input is restricted to be a cubic graph, unless P=NP\cite{alimonti2000some}. 
	By a simple reduction we will show that \textsc{Vertex Cover} still does not admit a PTAS when the input is restricted to be a graph that belongs to the class $\mathcal{Z}$, i.e.~a graph obtained from a cubic graph by subdividing every edge exactly 4 times. 
For an illustration we refer to Figure~\ref{fig:Cubic}, in which the first graph $K_4$ is the cubic graph on four vertices and the second graph is the graph obtained 
from $K_4$ by subdividing every edge 4 times.
	
	Given a cubic graph $G$, let $H \in \mathcal{Z}$ be  the graph obtained from $G$ by subdividing
	each edge 4 times. It is well known and can be easily verified that a double subdivision of an edge increases
	the size of a minimum vertex cover exactly by one. Hence, denoting by $\tau(G)$ and $\tau(H)$
	the sizes of minimum vertex covers of $G$ and $H$, respectively, we have:
	\begin{equation}\label{eq:vcGH}
		\tau(H) = \tau(G) + 2m,
	\end{equation}
	where $m$ is the number of edges in $G$.

Now we will prove that, starting from an arbitrary (not necessarily optimum) vertex cover $S_H$ for~$H$, 
we can efficiently compute a vertex cover $S_G$ for $G$ of size $|S_G| \leq |S_H| - 2m$. 
To this end, we show how to construct $S_G$ from $S_H$ by decreasing the cardinality of $S_H$ 
by at least two for every edge of $G$. 
Initially, we set $S_G = S_H$. 
Let $uv \in E(G)$ be an edge in $G$, and let $ua_1, a_1a_2,a_2a_3,a_3a_4,a_4v$ be the edges of $H$ 
that correspond to the 4-subdivision of $uv$. 
Note that $S_H$ contains at least two of the vertices $a_1, a_2, a_3, a_4$. 
Suppose that at least one of the vertices $u,v$ is contained in $S_H$. 
Then we just remove from $S_G$ all its vertices among $\{a_1, a_2, a_3, a_4\}$. 
Suppose otherwise that none of $u,v$ is contained in $S_H$. 
Then it is easy to verify that $S_H$ contains at least three of the vertices $a_1, a_2, a_3, a_4$. 
In this case, we add $u$ to $S_G$ and we remove from $S_G$ all its vertices among $\{a_1, a_2, a_3, a_4\}$. 
After repeating this procedure for every edge of $G$, we obtain a set $S_G$ that still covers all edges of $G$, while $|S_G| \leq |S_H| - 2m$ holds.

	To complete the proof, suppose for the sake of contradiction that there exists a PTAS 
	for \textsc{Vertex Cover} in $\mathcal{Z}$. 
	That is, for every $\epsilon >0$, we can compute in polynomial time a vertex cover $S_H$ of $H$
	such that $|S_H| \leq (1 + \epsilon)\tau(H)$. As we showed above, starting from $S_H$, 
	we can compute in polynomial time a vertex cover $S_G$ of $G$ such that
	\begin{equation}
		\begin{split}
			|S_G|  & \leq |S_H| - 2m \\
			& \leq (1 + \epsilon)\tau(H) - 2m \\
			& = (1 + \epsilon)(\tau(G) + 2m) - 2m \\
			& = (1 + \epsilon)\tau(G) + 2\epsilon \cdot m \\
			& \leq (1 + \epsilon)\tau(G) + 6\epsilon \cdot \tau(G) \\
			& = (1 + 7\epsilon)\tau(G),
		\end{split}
	\end{equation}
	where in the first equality we used (\ref{eq:vcGH}), and in the last inequality we used the fact
	that $m \leq 3\tau(G)$, because every vertex in the cubic graph $G$ covers exactly 3 edges.
	Summarizing, the existence of a PTAS for \textsc{Vertex Cover}  in the class $\mathcal{Z}$
	would imply the existence of a PTAS in the class of cubic graphs, which is a
	contradiction, unless P=NP~\cite{alimonti2000some}. 
\end{proof}

Let now $\mathcal{X}$ be the class of graphs whose connected components are 
induced subgraphs of graph $\Psi$ (see Figure 1). Clearly, \textsc{Vertex Cover} is linearly solvable on graphs from $\mathcal{X}$.
We will show that, unless P=NP, \DTVCproblemD{2} does not admit a PTAS on always $\mathcal{X}$ 
temporal graphs by using a reduction from \textsc{Vertex Cover} on $\mathcal{Z}$.

\begin{figure}[h]
\centering
\includegraphics[scale=0.9]{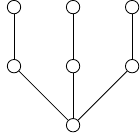}  
\caption{The graph $\Psi$.}
\label{fig:graphS}
\end{figure}

\begin{theorem}
\label{th:2DeltaSHard}  
\DTVCproblemD{2} on always $\mathcal{X}$ temporal graphs does not admit a PTAS, unless P=NP.
\end{theorem}

\begin{proof}
	To prove the theorem we will reduce \textsc{Vertex Cover} on $\mathcal{Z}$
	to \DTVCproblemD{2} on always $\mathcal{X}$ temporal graphs.
	Let $H = (V,E)$ be a graph in $\mathcal{Z}$.
	First we will show how to construct an always $\mathcal{X}$ temporal graph $(G,\lambda)$ of lifetime 2.
	Then we will prove that the size $\tau$ of a minimum vertex cover of $H$ is equal to the size $\sigma$
	of a minimum-cardinality \SWTVCD{2} of $(G, \lambda)$.  

	Let $R \subseteq V$ be the set of vertices of degree 3 in $H$.
	We define $(G, \lambda)$ to be a temporal graph of lifetime~2, where snapshot $G_1$ is obtained from $H$
	by removing the edges with both ends being at distance exactly~2 from $R$, and snapshot $G_2 = H - R$.
	Figure~\ref{fig:Cubic} illustrates the reduction for $H = K_4$.

	Let $\mathcal{S} = (S_1, 1) \cup (S_2, 2)$ be an arbitrary \SWTVCD{2} of $(G, \lambda)$ for some $S_1, S_2 \subseteq V$. 
	Since every edge of $H$ belongs to
	at least one of the graphs $G_1$ and $G_2$, the set $S_1 \cup S_2$ covers all the edges of $H$. 
	Hence, $\tau \leq |S_1 \cup S_2| \leq |S_1| + |S_2| = |\mathcal{S}|$. As $\mathcal{S}$ was chosen arbitrarily we further conclude that $\tau \leq \sigma$.

	To show the converse inequality, let $C \subseteq V$ be a minimum vertex cover of $H$. 
	Let $S_1$ be those vertices in $C$ which either have degree 3, or have a neighbor of degree 3. 
	Let also $S_2 = C \setminus S_1$. We claim that $(S_1, 1) \cup (S_2, 2)$ is a \SWTVCD{2} of $(G, \lambda)$.
	First, let $e \in E$ be an edge in $H$ incident to a vertex of degree 3. Then, by the construction, $e$ is active only in time slot 1,
	i.e.~$e \in E_1 \setminus E_2$, and a vertex $v$ in $C$ covering $e$ belongs to $S_1$. Hence, $e$ is temporally covered
	by $(v,1)$ in $(G, \lambda)$. Let now $e \in E$ be an edge in $H$ whose both end vertices have degree 2. 
	If one of the end vertices of $e$ is adjacent to a vertex of degree 3 in $H$, then, by the construction, $e$ is active in both
	time slots $1$ and $2$. Therefore, since $C = S_1 \cup S_2$, edge $e$ will be temporally covered in $(G, \lambda)$ in at least
	one of the time slots.
	Finally, if none of the end vertices of $e$ is adjacent to a vertex of degree 3 in $H$, then $e$ is active only in time slot~2, i.e.~$e \in E_2 \setminus E_1$. Moreover, by the construction a vertex $v$ in $C$ covering $e$ belongs to $S_2$.
	Hence, $e$ is temporally covered by $(v,2)$ in $(G, \lambda)$. This shows that $(S_1, 1) \cup (S_2, 2)$ is a 
	\SWTVCD{2} of $(G, \lambda)$, and thus $\sigma \leq |S_1| + |S_2| = |C| = \tau$. 
Therefore $\sigma = \tau$.
	
Note that for any $r\in \mathbb{N}$, any feasible solution of \DTVCproblemD{2} on $(G,\lambda)$ 
with size $r$ defines a vertex cover of $H$ with size at most $r$. 
Therefore, since $\sigma = \tau$ as we proved above, 
any PTAS for \DTVCproblemD{2} on always~$\mathcal{X}$ temporal graphs 
implies a PTAS for \textsc{Vertex Cover} in the class $\mathcal{Z}$, 
which is a contradiction by Lemma~\ref{lem:apxY}, unless P=NP. 
\end{proof}

\begin{figure}[h]
\centering
\includegraphics[scale=0.58]{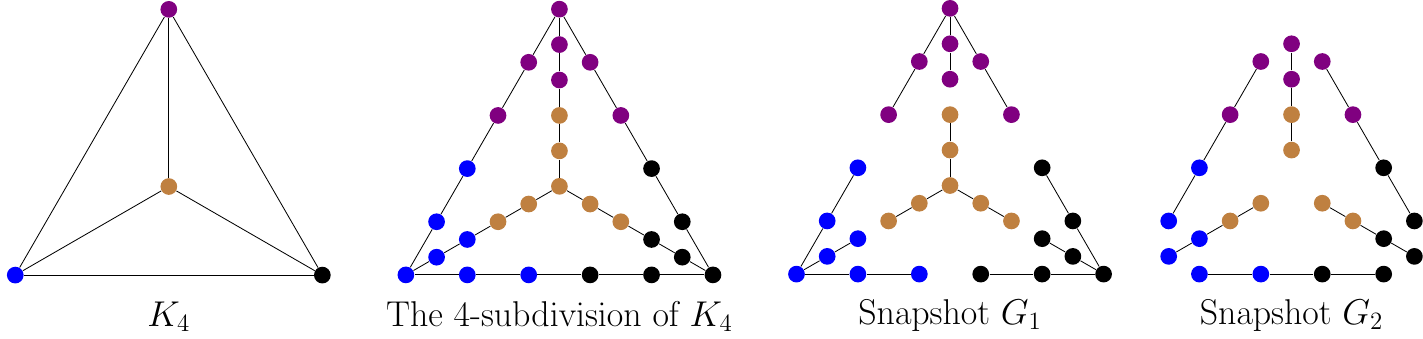}  
\caption{A cubic graph $K_4$, its 4-subdivision, and the corresponding
snapshots $G_1$ and $G_2$}
\label{fig:Cubic}
\end{figure}

	\section{Approximation algorithms}
	\label{sec:approx}
	
	In this section we provide several approximation algorithms for SW-TVC. 
	The approximation factors that we achieve depend on various parameters of the input temporal graph. 
	The values of these parameters (and of the corresponding approximation factors) are in general 
	incomparable to each other, and thus the best option for approximating the optimal solution depends 
	on the values of those parameters in each specific input instance.

	\subsection{Approximations in terms of $T$, $\Delta$, and the largest edge frequency}
	\label{subsec-approx-T-Delta-Freq}

	We begin by presenting a reduction from \SWTVCProblemShort\ 
	to \textsc{Set Cover}, which proves useful for deriving approximation
	algorithms for the original problem.
	We note here the similarity of the construction in this reduction to the construction in the reduction 
	of \textsc{TVC} to \textsc{Set Cover} presented in 	Theorem~\ref{TVC-general-approximation}.
	Consider an instance, $(G, \lambda)$ and $\Delta\leq T$, of the
	\SWTVCProblemShort\ problem. Construct an instance of \textsc{Set Cover} as
	follows:
	Let the universe be $U= \{ (e,t) :  e \in E[W_t], t \in [1,T - \Delta + 1] \}$,
	i.e.~the set of all pairs $(e,t)$ of an edge $e$ and a time slot $t$ 
	such that $e$ appears (and so must be temporally covered) within window $W_t$. 
	For every vertex appearance $(v,s)$ we define $C_{v,s}$ to be
	the set of elements $(e,t)$ in the universe $U$, such that $(v,s)$ temporally covers $e$ in
	the window $W_t$. 
	Formally, 
	$C_{v,s} = \{ (e,t) :  v \text{ is an endpoint of } e, e\in E_{s}, \text{ and }
	s \in W_t \}$.
	Let $\mathcal{C}$ be the family of all sets $C_{v,s}$, where $v \in V, s \in [1,T]$.
	The following lemma
	shows that finding a minimum-cardinality \SWTVC\ of $(G,\lambda)$ 
	is equivalent to finding a minimum-cardinality family of sets $C_{v,s}$ that
	covers the universe $U$.
	
	\begin{lemma}\label{cl:toSetCoverReduction}
		A family $\mathcal{C} = \{ C_{v_1,t_1}, \ldots, C_{v_k,t_k} \}$ is a set cover
		of $U$ if and only if 
		$\mathcal{S} = \{ (v_1,t_1), \ldots, (v_k,t_k) \}$ is a \SWTVC\ of $(G,
		\lambda)$.
	\end{lemma}
	\begin{proof}
		First assume that $\mathcal{C}$ is a set cover of $U$, but $\mathcal{S}$
		is not a \SWTVC\ of $(G, \lambda)$. Then there exists a window $W_r$ for some
		$r \in [1,T-\Delta+1]$, such that an edge
		$uv$ appears in $W_r$ but is not temporally covered in $W_r$ by
		$\mathcal{S}$.
		This means that, for every $j \in W_r$ such that $uv \in E_j$, neither $(u,j)$ nor $(v,j)$ 
		belongs to $\mathcal{S}$.
		Therefore, $C_{u,j}, C_{v,j} \notin \mathcal{C}$ for all $j \in W_r$ such that
		$uv \in E_j$.
		But then $\mathcal{C}$ does not cover $(uv,r) \in U$, which is a
		contradiction.
		
		Conversely, assume that $\mathcal{S}$ is a \SWTVC\ of $(G, \lambda)$, but
		$\mathcal{C}$ is not a set cover of $U$, i.e.~there exists some $(uv,r) \in
		U$
		which belongs to none of the sets in $\mathcal{C}$. The latter means that 
		$C_{u,j}, C_{v,j} \notin \mathcal{C}$, and therefore $(u,j), (v,j) \notin
		\mathcal{S}$, for every
		$j \in W_r$ such that $uv \in E_j$. Therefore $uv$ is not covered in $W_r$,
		which is a contradiction.
	\end{proof}
	
	\begin{description}
		\item[$(\ln{n} + \ln{\Delta} + \frac{1}{2})$-approximation.]
		In the instance of \textsc{Set Cover} constructed by the above reduction, every set $C_{v,s}$ in $\mathcal{C}$ contains at most $n\Delta$ elements of the universe $U$. 
		Indeed, the vertex appearance $(v,s)$ temporally covers at most $n-1$ edges, each in at most $\Delta$ windows (namely from window $W_{s-\Delta+1}$ up to window $W_s$). 
		Thus we can apply the polynomial-time greedy algorithm of~\cite{DuhFurer97} for \textsc{Set Cover} which achieves an approximation ratio of $H_{n\Delta}-\frac{1}{2}$, 
		where $n\Delta$ is the maximum size of a set in the input instance and $H_{n\Delta}=\sum_{i=1}^{n\Delta}\frac{1}{i} \leq \ln{n} + \ln{\Delta} + 1$ is the $n\Delta$-th harmonic number.

		\item[$2k$-approximation, where $k$ is the maximum edge frequency.]
		Given a temporal graph $(G,\lambda)$ and an edge $e$ of $G$, the $\Delta$-frequency
		of $e$ is the maximum number of time slots at which $e$
		appears within a $\Delta$-window. Let $k$ denote the maximum
		$\Delta$-frequency over
		all edges of $G$.
		Clearly, for a particular $\Delta$-window $W_t$, an edge $e \in E[W_t]$ can be
		temporally 
		covered in $W_t$ by at most $2 k$ vertex appearances.
		So in the above reduction to \textsc{Set Cover}, 
		every element $(e,t) \in U$ belongs to at most $2 k$ sets
		in $\mathcal{C}$. Therefore, the optimal solution of the constructed
		instance of \textsc{Set Cover}
		can be approximated within a factor of $2k$ in polynomial time
		\cite[p.~118-9]{vazirani2003approximation}, yielding a polynomial-time $2k$-approximation
		for \SWTVCProblemShort.
		
		\item[$2\Delta$-approximation.]
		Since the maximum $\Delta$-frequency of an edge is always upper-bounded by $\Delta$, the 
		previous algorithm gives a worst-case polynomial-time $2\Delta$-approximation for \SWTVCProblemShort\ on arbitrary temporal graphs.
		
	\end{description}

	\subsection{Approximation in terms of maximum degree of snapshots}
	\label{subsec-approx-max-degree-snap}
	
	In this section we give a polynomial-time $d$-approximation algorithm for the
	\SWTVCProblemShort\ problem on \emph{always degree at most $d$} temporal graphs, that is, on temporal graphs where the maximum degree in each snapshot is at most $d$.
	In particular, the algorithm computes an optimum solution (i.e.~with approximation ratio $d=1$) for always matching (i.e.~always degree at most $1$) temporal graphs. 
	As a building block, we first provide an exact $O(T)$-time algorithm for optimally solving \SWTVCProblemShort\ in the class of single-edge temporal graphs, namely temporal graphs whose underlying graph is a single edge.
To that end, we reduce \SWTVCProblemShort\ to \textsc{Interval Covering}, 
leading to an intuitive algorithm which selects the ``rightmost'' appearance 
of the edge of the temporal graph within a time window in which the edge has 
not been covered yet (starting from the first time window). In fact, this 
algorithm is a direct translation of a known greedy algorithm which solves 
\textsc{Interval Covering} in the \SWTVCProblemShort\ single-edge-temporal-graphs setting.
Once we have established this exact algorithm for single-edge temporal graphs, 
we prove that for always degree at most $d$ temporal graphs we can $d$-approximate 
the optimal solution by independently solving \textsc{SW-TVC} for every single-edge temporal subgraph 
and then taking the union of these solutions.

	\paragraph*{Single-edge temporal graphs}	
	Consider a temporal graph $(G_{0}, \lambda)$ where $G_{0}$ is the single-edge graph, i.e.~$V(G_{0}) = \{ u, v \}$ and $	E(G_{0}) = \{ uv \}$. We reduce \SWTVCProblemShort\ on $(G_{0}, \lambda)$ to an
	instance of	\textsc{Interval Covering}, which has a known greedy algorithm that we then translate to an algorithm for \SWTVCProblemShort\ on single-edge temporal graphs.
	
	\vspace{0,1cm} \noindent \fbox{ 
		\begin{minipage}{0.96\textwidth}
			\begin{tabular*}{\textwidth}{@{\extracolsep{\fill}}lr} \textsc{Interval Covering} & \\ \end{tabular*}
			
			\vspace{1.2mm}
			{\bf{Input:}}  A family $\mathcal{I}$ of intervals in the line.\\
			{\bf{Output:}} A minimum-cardinality subfamily $\mathcal{I}' \subseteq \mathcal{I}$
			such that $\bigcup_{I \in \mathcal{I}} = \bigcup_{I \in \mathcal{I}'}$.
	\end{minipage}} \vspace{0,3cm}

	We construct the family $\mathcal{I}$ as follows. For every $i =1,2, \ldots, T$ such that $uv \in E_i$ 
	we include into $\mathcal{I}$ the interval $I_i = [i - \Delta + 1, i] \cap
	[1,T-\Delta+1]$,
	which contains the first time slot 
	of each of those $\Delta$-windows that include time slot $i$.
	
	\begin{lemma}\label{cl:IntervalCoveringReduction}
		Let $i_1, i_2, \ldots, i_k$ be such that $uv \in E_{i_j}$ for every $j= 1,2, \ldots, k$.
		Then $\mathcal{I}' = \{ I_{i_1}, \ldots, I_{i_k} \}$ is an interval covering
		of
		$\mathcal{I}$ if and only if
		$\mathcal{S} = \{ (u,i_1), (u,i_2), \ldots, (u,i_k) \}$ is a \SWTVC\ of
		$(G_{0}, \lambda)$.
	\end{lemma}
	\begin{proof}
		Assume first that $\mathcal{I}'$ is an interval covering of $\mathcal{I}$, but
		$\mathcal{S}$ is not
		a \SWTVC\ of $(G_{0}, \lambda)$. The latter means that there exists a
		$\Delta$-window $W_t$
		such that $uv$ exists at some time slot $s$ in $W_t$, but $uv$ is not
		temporally covered by
		any vertex appearance in $\mathcal{S}[W_t]$. Therefore $t \notin I_{i_1} \cup
		I_{i_2} \cup \ldots \cup I_{i_k}$,
		but $t \in I_s$, which contradicts the assumption that $\mathcal{I}'$ is an
		interval covering of $\mathcal{I}$.
		
		Conversely, assume that $\mathcal{S}$ a \SWTVC\ of $(G_{0}, \lambda)$, but
		$\mathcal{I}'$
		is not an interval covering of $\mathcal{I}$, that is, there exists  $I_i \in
		\mathcal{I}$ and $t \in I_i$
		such that $t \notin \bigcup_{I \in \mathcal{I}'} I$. By the construction, this
		means that $uv \in E_i$ is not
		temporally covered by any vertex appearances in $\mathcal{S}[W_t]$, which is a contradiction.
	\end{proof}

	Lemma~\ref{cl:IntervalCoveringReduction} shows that finding a minimum-cardinality \SWTVC\ 
	of~$(G_{0},\lambda)$ is equivalent to
	finding a minimum interval covering of $\mathcal{I}$. An easy linear-time greedy algorithm
	for the \textsc{Interval Covering} picks at each iteration, among the intervals that cover the leftmost uncovered point, the one with largest finishing time.
	Algorithm~\ref{alg:1edgeVC} implements this simple rule in the context of the
	\SWTVCProblemShort\ problem.

	\begin{algorithm}[h]
		\caption{\SWTVCProblemShort\ on single-edge temporal graphs}
		\label{alg:1edgeVC}
		\begin{algorithmic}[1]
			\REQUIRE{A temporal graph $(G_{0}, \lambda)$ of lifetime $T$ with $V(G_{0}) = \{ u,v\}$, and $\Delta \leq T$.}
			\ENSURE{A minimum-cardinality \SWTVC\ $\mathcal{S}$ of $(G_{0}, \lambda)$.}

			\STATE{$\mathcal{S} \leftarrow \emptyset$}
			
			\STATE{$t = 1$}
			
			\WHILE{$t \leq T- \Delta +1$}
			\IF{$\exists r \in [t, t + \Delta -1]$ such that $uv \in E_r$}
			\STATE{choose maximum such $r$ and add $(u,r)$ to $\mathcal{S}$}\label{line:choice_of_r}
			\STATE{$t \leftarrow r+1$}
			\ELSE
			\STATE{$t \leftarrow t+1$}
			\ENDIF
			\ENDWHILE
			
			\RETURN{$\mathcal{S}$}
		\end{algorithmic}
	\end{algorithm}

	\begin{lemma}
		Algorithm~\ref{alg:1edgeVC} solves \SWTVCProblemShort\ on a single-edge temporal graph and can be implemented to work in time $O\left( T \right)$. 
	\end{lemma}
	\begin{proof}
		The time complexity of the algorithm is dominated by the running time of the while-loop. We provide an implementation of the while-loop, which works in time $O(T)$: in each iteration, we inspect the current $\Delta$-window $[t, t+\Delta -1]$ from the rightmost time slot moving to the left. As we go through the time slots, we mark the ones in which edge $uv$ does not appear as ``NO''. When we move to the next iteration (the next $\Delta$-window), we do not need to revisit any time slots that have been marked as ``NO'' and we immediately move to the next iteration whenever we meet such a slot. This way we visit every time slot at most once, and hence we exit the while-loop after $O(T)$ operations.
	\end{proof}

	\paragraph*{Always degree at most $d$ temporal graphs}

	We present now the main algorithm of this section, the idea of which is to independently solve \SWTVCProblemShort\ for every possible single-edge temporal subgraph 
	of a given temporal graph by Algorithm~\ref{alg:1edgeVC}, and take the union of these solutions.
	We will show that this algorithm is a $d$-approximation algorithm for \SWTVCProblemShort\ on always degree at most $d$ temporal graphs.
	
	Let $(G, \lambda)$ be a temporal graph,  where $G=(V,E)$, $|V|=n$, and $|E|=m$. For every edge $e=uv \in E$, let $(G[\{u,v\}], \lambda)$ 
	denote the temporal graph where the underlying graph is the induced subgraph $G[\{u,v\}]$ of $G$ and the labels of $e$ are exactly the same as in $(G, \lambda)$.
	
	\begin{algorithm}[H]
		\caption{$d$-approximation of \SWTVCProblemShort\ on always degree at most $d$ temporal graphs} 
		\label{alg:alwaysDapprox}
		\begin{algorithmic}[1]
			\REQUIRE{An always degree at most $d$ temporal graph $(G, \lambda)$ of lifetime $T$, and $\Delta \leq T$.}
			\ENSURE{A \SWTVC\ $\mathcal{S}$ of $(G, \lambda)$.}
			
			\FOR{$i=1$ to $T$} 
			\STATE{$\mathcal{S}_i \leftarrow \emptyset$}
			\ENDFOR
			
			\FOR{every edge $e=uv \in E(G)$} \label{line:DapproxLine1}
			\STATE{Compute the optimal solution $\mathcal{S}^{e}$ of the problem for
				$(G[\{u,v\}], \lambda)$
				by Algorithm~\ref{alg:1edgeVC}}\label{line:DapproxLine}
			
			\FOR{$i=1$ to $T$} 
			\STATE{$\mathcal{S}_i \leftarrow \mathcal{S}_i \cup \mathcal{S}^{e}_i$}
			\ENDFOR
			
			\ENDFOR\label{line:DapproxLine2}
			
			\RETURN{$\mathcal{S}$}
		\end{algorithmic}
	\end{algorithm}

	\begin{lemma}\label{lem:alwaysDapprox}
		Algorithm~\ref{alg:alwaysDapprox} is a $O\left( m T \right)$-time $d$-approximation algorithm for
		\SWTVCProblemShort\ on always degree at most 
		$d$ temporal graphs.
	\end{lemma}
	\begin{proof}
		Let $(G, \lambda)$ be an always degree at most $d$ temporal graph of
		lifetime $T$, and
		let $\mathcal{S}^*$ be a minimum-cardinality \SWTVC\ of $(G, \lambda)$.
		We will show that $|\mathcal{S}| \leq d \cdot |\mathcal{S}^*|$, where $\mathcal{S}$ is the solution 
		computed by Algorithm~\ref{alg:alwaysDapprox}.
		To this end we apply a double counting argument to the set~$C$ of all triples
		$(v,e,t) \in V \times E \times [1,T]$
		such that $v \in \mathcal{S}^*_t$, $e \in E_t$, and $v$ is incident to $e$.
		
		On the one hand 
		$$
		|C| = \sum\limits_{t=1}^{T} \sum\limits_{v \in \mathcal{S}^*_t} | \{
		(v,e,t) : e \in E_t \text{ and $v$ is incident to $e$} \} |
		\leq \sum\limits_{t=1}^{T} \sum\limits_{v \in \mathcal{S}^*_t} d = d \cdot
		|\mathcal{S}^*|,
		$$
		where the inequality follows from the assumption that every snapshot of $(G, \lambda)$ has
		maximum degree at most $d$.
		
		On the other hand
		$$
		|C| = \sum\limits_{e \in E} \sum\limits_{t =1}^{T} |\{ (v,e,t) :  e \in E_t, v
		\in \mathcal{S}^*_t,
		\text{ and $v$ is incident to $e$} \}|
		\geq \sum\limits_{e \in E}
		|\mathcal{S}^{e}| = |\mathcal{S}|,
		$$
		where $\mathcal{S}^{e}$ is the optimum \SWTVC\ of the temporal graph $(G[\{u,v\}], \lambda)$ 
		(see line~\ref{line:DapproxLine} of Algorithm~\ref{alg:alwaysDapprox}). 
		The last inequality follows from the fact that the restriction of any
		\SWTVC\ of $(G, \lambda)$ to the temporal subgraph induced by the endpoints of $e$ is a \SWTVC\
		of the temporal subgraph, and therefore
		has cardinality at least~$|\mathcal{S}^{e}|$. 		
		We conclude that $|\mathcal{S}| \leq |C| \leq d \cdot |\mathcal{S}^*|$, as required. 
		
		The time-complexity of Algorithm~\ref{alg:alwaysDapprox} is dominated by the time needed to execute the for-loop of lines~\ref{line:DapproxLine1}-\ref{line:DapproxLine2}. The latter requires time $O\left(mT\right)$.	
	\end{proof}

	Note that in the case of always matching temporal graphs, i.e.~where every snapshot is a matching, the maximum degree in each snapshot is $d=1$, so the above $d$-approximation actually yields an exact algorithm (see Corollary~\ref{cor:always_mathcing}).
	This is not surprising, since a single vertex appearance can only cover one edge in always matching temporal graphs. Therefore, the solutions for different edges can be independently optimized.
	\begin{corollary}\label{cor:always_mathcing}
		\SWTVCProblemShort\ can be optimally solved in $O(mT)$
		time on the class of always matching temporal graphs.
	\end{corollary}

	\section{Conclusions and open problems}
	\label{sec-conclusions}

	In this paper we introduced and studied two natural temporal extensions of the problem \textsc{Vertex Cover} 
	for static graphs, namely \textsc{Temporal Vertex Cover} and \SWTVCProblem\ (for short, \textsc{TVC} and \SWTVCProblemShort, respectively), which take into account the dynamic nature of temporal networks.
	We presented a thorough investigation of the complexity and approximability of these problems, including
	strong hardness results, and various approximation and exact algorithms.
	In particular, for \SWTVCProblemShort\ we designed a linear-time optimal algorithm on always degree at most 1 temporal graphs, i.e.~on temporal graphs that consist of a matching in each time step. 
On the other hand, we showed that \SWTVCProblemShort\ becomes NP-hard on always degree at most 3 temporal
	graphs, even when the underlying graph is cubic and every snapshot has connected components with 
at most~7 vertices. This leaves an intriguing open question of the complexity status of the problem on always degree at most~2
	temporal graphs, i.e.~on temporal graphs that consist of disjoint paths and cycles in each time step.
	
	\begin{problem}
	Restricted on always degree at most 2 temporal graphs, is \SWTVCProblemShort\ efficiently solvable?
	\end{problem}
	
	\noindent
	For \SWTVCProblemShort\ we provided various polynomial-time approximation algorithms, including a 
	$2\Delta$-approximation algorithm for general temporal graphs, and a $d$-approximation algorithm for
	always degree at most $d$ temporal graphs. There are no known matching lower bounds for these approximation factors, and it would be interesting to know whether there is any room for improvement.
	
	\begin{problem}
		Let $d$ be the maximum vertex degree in each snapshot of a temporal graph. 
		Can $\SWTVCProblemShort$ be efficiently approximated within a factor better than $2\Delta$ or better than $d$? 
		In particular, does there exist a polynomial-time approximation algorithm for $\SWTVCProblemShort$ with approximation factor $o(\Delta)$ or $o(d)$?
	\end{problem}

\noindent
A natural extension of the problem \textsc{SW-TVC} is computing a sliding $\Delta$-window temporal vertex cover 
that minimizes the maximum cost (i.e.~the maximum number of vertex appearances) over all time windows of a given length $\ell\geq \Delta$. Formally this problem can be defined as follows.

\vspace{0,3cm} \noindent \fbox{ 
\begin{minipage}{0.96\textwidth}
 \begin{tabular*}{\textwidth}{@{\extracolsep{\fill}}lr} \textsc{Restricted-Length Sliding Window Temporal Vertex Cover} \ \ (\textsc{RL-SW-TVC}) & \\ \end{tabular*}
 
  \vspace{1.2mm}
{\bf{Input:}}  A temporal graph $(G,\lambda)$ with lifetime $T$, and two integers $\Delta \leq \ell \leq T$.\\
{\bf{Output:}} A sliding $\Delta$-window temporal vertex cover $\mathcal{S}$ of $(G,\lambda)$ 
such that the number of vertex appearances of $\mathcal{S}$ in any time window of length $\ell$ is minimized.
\end{minipage}} \vspace{0,3cm}

Clearly, \textsc{SW-TVC} is the special case of \textsc{RL-SW-TVC} where $\ell=T$, in which case we aim at minimizing the total number of vertex appearances in the solution. 
Although the two problems might look superficially similar to each other, they are different, 
as we illustrate in the example of Figure~\ref{RL-SW-TVC-example-fig}, where $T=3$ and $\ell=\Delta=2$. 
Figure~\ref{RL-SW-TVC-example-fig-A} an optimal solution to \textsc{SW-TVC} is highlighted, which 
contains in total 4 vertex appearances (all at time $t=2$). In this solution, the maximum number of vertex appearances in every time window of length $\ell=2$ is 4. 
On the other hand, Figure~\ref{RL-SW-TVC-example-fig-B} an optimal solution to \textsc{RL-SW-TVC} 
for $\ell=2$ is highlighted, which contains in total 5 vertex appearances. 
In this solution, the maximum number of vertex appearances in every time window of length $\ell=2$ is 3.

	\begin{problem}
	Is the problem \textsc{RL-SW-TVC} strictly harder than \textsc{SW-TVC}? Can our results be extended to \textsc{RL-SW-TVC}?
	\end{problem}

\begin{figure}[h]
\centering%
\subfigure[]{ \label{RL-SW-TVC-example-fig-A}
\includegraphics[width=0.48\textwidth]{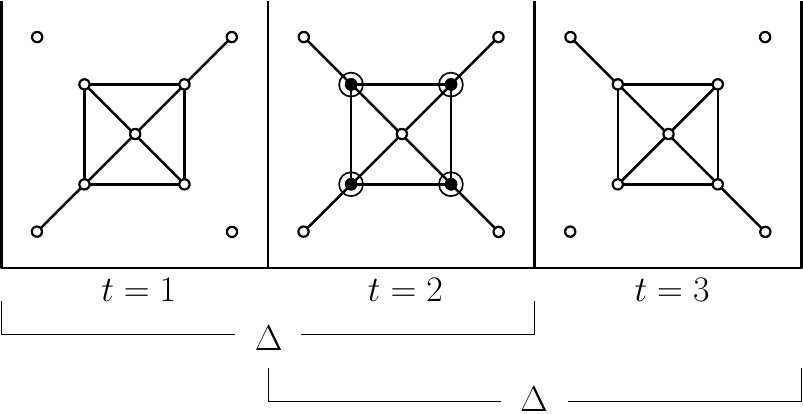}} \hspace{0,195cm} 
\subfigure[]{ \label{RL-SW-TVC-example-fig-B}
\includegraphics[width=0.48\textwidth]{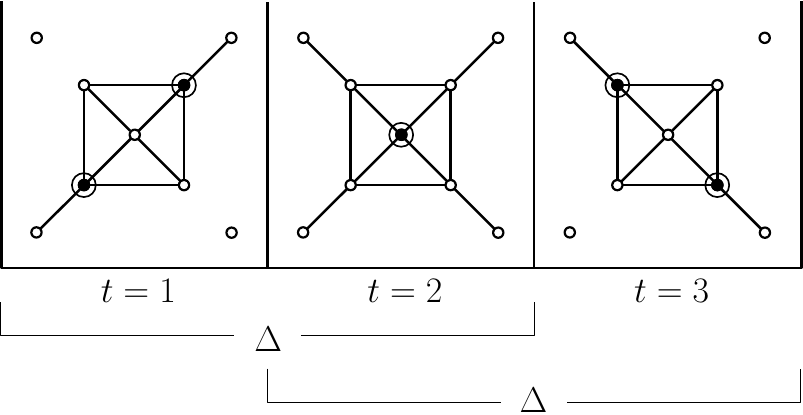}}
\caption{A temporal graph $(G,\lambda)$ with lifetime $T=3$, where $\Delta=2$. 
(a) An optimal solution to \textsc{SW-TVC} and (b)~an optimal solution to \textsc{RL-SW-TVC} where $\ell=\Delta=2$.}
\label{RL-SW-TVC-example-fig}
\end{figure}

\noindent
Finally, it will be interesting to investigate the practical aspects of the studied problems, 
\textsc{TVC} and \SWTVCProblemShort, and in particular to practically evaluate the performance of the different presented algorithms on real-world instances.

	\begin{problem}
	How do the different approximation and exact algorithms perform on real-world temporal graph instances?
	\end{problem}

\end{document}